\documentclass[journal]{IEEEtran}
\PassOptionsToPackage{numbers, compress}{natbib}
\usepackage{tikz}
\usetikzlibrary{calc}
\usepackage[T1]{fontenc}
\usepackage[utf8]{inputenc}
\usepackage{amsmath}
\usepackage{amsthm, mathtools}
\usepackage{amssymb}
\usepackage{graphicx, hyperref}

\usepackage{algorithm}
\usepackage{algpseudocode}

\makeatletter

\theoremstyle{plain}
\newtheorem{theorem}{\protect\theoremname}
\theoremstyle{plain}
\newtheorem{corollary}[theorem]{\protect\corollaryname}

\theoremstyle{plain}
\newtheorem{proposition}[theorem]{\protect\propositionname}
\usepackage[dvipsnames]{xcolor}

\def\calX{{\mathcal X}}
\def\calY{{\mathcal Y}}
\def\calC{{\mathcal C}}
\def\calS{{\mathcal S}}
\def\calA{{\mathcal A}}
\def\calB{{\mathcal B}}
\def\supp{\mathop{\rm supp}}
\def\Unif{\mathrm{Unif}}
\def\Ber{\mathrm{Bern}}
\def\TV{\mathrm{TV}}

\usepackage{titlesec}

\titleformat{\paragraph}[runin]
  {\normalfont\normalsize\bfseries} % Format
  {}                               % Label (Empty = No Number)
  {0pt}                            % Space between label and title
  {}                               % Code before

\titlespacing*{\paragraph}
  {0pt}                            % Indent
  {1.5ex plus 0.5ex minus .2ex}    % Space above
  {1em}                            % Space after (Horizontal)

\makeatother

\usepackage{babel}
\providecommand{\corollaryname}{Corollary}
\providecommand{\propositionname}{Proposition}
\providecommand{\lemmaname}{Lemma}
\providecommand{\theoremname}{Theorem}

\title{Learning from Totally Unpaired Data via Minimum Conditional Entropy Coupling}

\usepackage{lipsum}

\makeatletter
\def\blfootnote{\gdef\@thefnmark{}\@footnotetext}
\makeatother

\author{%
Cheuk Ting Li${}^\S$, 
Jun Chen${}^\dagger$, and 
Shahab Asoodeh${}^\ddagger$
\\

\thanks{${}^\S$C. T. Li is with the Department of Information Engineering, The Chinese University of Hong Kong, Hong Kong, China (email: ctli@ie.cuhk.edu.hk).}
\thanks{${}^\dagger$J. Chen is with the Department of Electrical and Computer Engineering, McMaster University, Hamilton, ON L8S 4K1, Canada (email: {chenjun@mcmaster.ca}).}
\thanks{${}^\ddagger$S. Asoodeh is with the Department of Computing and Software, McMaster University, Hamilton,
ON L8S 4K1, Canada (email: {asoodeh@mcmaster.ca}).}

}

\date{}

\begin{document}

\onecolumn

\maketitle

\begin{abstract}
Consider an ``unsupervised classification'' setting, where the goal is to learn a predictor of the categorical class variable $X$ using the categorical feature variables $Y_1,\ldots,Y_m$, but we can only train the predictor on $m+1$ totally unpaired datasets, where each dataset only consists of unordered samples of one of the variables $X,Y_{1},\ldots,Y_{m}$. Although this problem is highly underdetermined, we demonstrate that, under the assumption that a predictor with small cross entropy loss exists, we can find a good predictor by coupling the marginal distributions of $X,Y_1,\ldots,Y_m$ in a manner that minimizes the conditional entropy $H(X|Y_1,\ldots,Y_m)$. This problem, which we call minimum conditional entropy coupling (MCEC), is a generalization of the widely studied minimum entropy coupling problem (where $m=1$). Although the MCEC problem is NP-hard, we give an algorithm with a time complexity that is polynomial in $N$, which gives a coupling with conditional entropy within $O(\log \log N)$ from the optimum, where $N$ is the dimension of the output coupling. We also evaluate the performance of the MCEC predictor on synthetic and real-world datasets.
\end{abstract}

\section{Introduction}

Consider a classical classification setting, where we predict the categorical class variable $X$ using the categorical feature variables $Y_{1},\ldots,Y_{m}$. The model maps $Y_{1},\ldots,Y_{m}$ to a predicted distribution $\hat{P}(x|Y_{1},\ldots,Y_{m})$ of $X$. We find the model that minimizes the average cross entropy loss $\mathbb{E}[-\log\hat{P}(X|Y_{1},\ldots,Y_{m})]$, where the expectation is over the empirical distribution of the dataset. If the model is rich enough to represent any conditional distribution from $Y_{1},\ldots,Y_{m}$ to $X$, then the optimal average loss is the conditional entropy $H(X|Y_{1},\ldots,Y_{m})$ evaluated over the empirical distribution.

To evaluate the cross entropy loss, we require a paired, labeled dataset, where each entry is a tuple $(X,Y_{1},\ldots,Y_{m})$. Is learning still possible if the dataset is \emph{totally unpaired}, that is, we only have $m+1$ separate datasets, one of them being an unordered collection of values of $X$, another being an unordered collection of values of $Y_{1}$, etc.? Equivalently, we are given the marginal empirical distribution $P$ of $X$, and the marginal empirical distributions $Q_{i}$ of $Y_{i}$. Learning from unpaired data has been widely studied for image-to-image translation \cite{zhu2017unpaired,liu2017unsupervised,kim2017learning,almahairi2018augmented} and single-cell alignment \cite{welch2017matcher,amodio2018magan,bounoua2025learning}.  Unlike previous works which focus on unsupervised learning, we consider a seemingly impossible task of classification (usually considered a supervised learning task) with only unsupervised training data (the class of each data point is unknown).

If the dataset is completely unlabeled and unpaired, how is it possible to learn the dependency between $X$ and $Y_{1},\ldots,Y_{m}$? They are general categorical variables, so we cannot impose domain-specific assumptions (e.g., the shared-latent space assumption \cite{liu2017unsupervised} for image-to-image translation). This is a highly underdetermined problem, which is nevertheless possible under the following optimistic assumption:

\noindent \textbf{Predictability assumption:} A predictor with small cross entropy loss indeed exists, i.e., $H(X|Y_{1},\ldots,Y_{m})$ for the true distribution is small.

The rationale behind this assumption is simple: if our goal is to learn a good predictor, then we should ignore the situation where such a predictor does not exist since our goal is unattainable anyway. To learn a predictor under this assumption, we find the best way to pair these datasets, or couple these distributions, which gives the smallest $H(X|Y_{1},\ldots,Y_{m})$.  The optimal value is the following quantity 
\begin{equation}
H^{*}(P\Vert Q_1, \ldots, Q_m):=\! \inf_{X\sim P,Y_{1}\sim Q_{1},\ldots,Y_{m}\sim Q_{m}}\!\! H(X|Y_{1},\ldots,Y_{m}),\label{eq:mcec}
\end{equation}
where the infimum is over couplings $(X,Y_{1},\ldots,Y_{m})$ of $P,Q_{1},\ldots,Q_{m}$. In the homogeneous setting when $Q_i = Q$, we denote the above optimization by $H^{*}_m(P\Vert Q)$, that is
\begin{equation}
H^{*}_m(P\Vert Q):=\! \inf_{X\sim P,Y_{1},\ldots,Y_{m}\sim Q} H(X|Y_{1},\ldots,Y_{m}),\label{eq:mcec}
\end{equation}
The optimal predictor, which we call \emph{minimum conditional entropy coupling (MCEC) predictor}, would then be the conditional distribution $P_{X|Y_{1},\ldots,Y_{m}}(x|y_{1},\ldots,y_{m})$ computed from the optimal coupling. The MCEC predictor is the ``best predictor in the best case,'' in the sense that it attains the lowest best-case cross entropy loss, where ``best-case'' refers to the best coupling given the marginals.\footnote{Refer to Section~\ref{sec:limitations} for the reason we consider the best case instead of the worst case.} When $m=1$, we minimize $H(X|Y_{1})$, or equivalently, minimize $H(X,Y_{1})=H(X|Y_{1})+H(Q_{1})$, which is known as the \emph{minimum entropy coupling problem} \cite{vidyasagar2012metric,painsky2013memoryless,kovavcevic2015entropy,kocaoglu2017entropic,kocaoglu2017greedy,cicalese2019minimum}. MCEC generalizes minimum entropy coupling, and gives new applications that apply to minimum entropy coupling ($m=1$) as well.\footnote{We emphasize that $m=1$ does not only refer to the case where there is only one feature variable, but is also applicable to the scenario where there are multiple features, and we have two datasets: one with all features (but without class labels), and one with the class labels. In this case, all features can be grouped into $Y_{1}$. Also, the special case $m=1$ was used in a different setting of matching unpaired data in \cite{bounoua2025learning}. Refer to Section \ref{sec:previous}.} 

We may also make an additional optional assumption:

\noindent\textbf{Independent features assumption} $Y_{1},\ldots,Y_{m}$ are mutually independent. Equivalently, $H(Y_{1},\ldots,Y_{m})$ is maximized.

This is the ``opposite'' of the assumption of the naive Bayes classifier \cite{domingos1997optimality} that $Y_{1},\ldots,Y_{m}$ are conditionally independent given $X$.\footnote{The naive Bayes classifier \cite{domingos1997optimality} assumes a Bayesian network where there is an edge from $X$ to each $Y_{i}$.
%($X$ is a ``latent cause'' of $Y_{i}$'s). 
The independent features assumption is the opposite Bayesian network where there is an edge from each $Y_{i}$ to $X$. %($X$ is an ``effect'' of $Y_{i}$'s).
} If we make both the predictability and the independent features assumptions, we should solve
\begin{equation}
H^{*}(P\Vert\otimes_{i}Q_{i})=\inf_{X\sim P,(Y_{1},\ldots,Y_{m})\sim\otimes_{i}Q_{i}}H(X|Y_{1},\ldots,Y_{m}),\label{eq:imec}
\end{equation}
where $\otimes_{i}Q_{i}=Q_{1}\otimes\cdots\otimes Q_{m}$ (the product distribution). We call the resulting predictor the \emph{independent minimum entropy coupling (IMEC) predictor}. Note that $H^{*}(P\Vert\otimes_{i}Q_{i})$ is simply (\ref{eq:mcec}) substituted with $m\leftarrow1$, $Q_{1}\leftarrow\otimes_{i}Q_{i}$. Hence, (\ref{eq:imec}) is equivalent to the minimum entropy coupling problem between $P$ and $\otimes_{i}Q_{i}$.

To demonstrate the usefulness of the MCEC predictor, consider the following motivating example. We have the results of three separate surveys on mobile devices: a survey on operating systems (OS) showing that $55\%$ of the devices run OS A and $45\%$ run OS B; a survey on radio access technology (RAT) showing that $95\%$ run 5G and $5\%$ run 4G; and a survey on a certain bug showing that the bug is present in $52\%$ of the devices. These data are unpaired, and we do not know, for example, the portion of devices which run OS A and have the bug. Still, we would attempt to build a predictor which predicts whether a device has the bug based on its OS and RAT. This appears to be an ill-posed problem.  Nevertheless, if we make the predictability assumption, then we can solve (\ref{eq:mcec}) to find an optimal coupling ($52\%$ OS A 5G with bug, $3\%$ OS A 4G without bug, $43\%$ OS B 5G without bug, $2\%$ OS B 4G without bug), and a predictor that predicts that the bug is present in OS A devices running 5G, which is perfectly accurate over the optimal coupling.\footnote{There is another optimal coupling ($51\%$ OS A 5G w/ bug, $4\%$ OS A 4G w/o bug, $44\%$ OS B 5G w/o bug, $1\%$ OS B 4G w/ bug) and predictor (bug in OS A 5G and OS B 4G). These two are the only optimal couplings. Even though optimal coupling is ambiguous, a predictor built using one optimal coupling has an accuracy at least $98\%$ over the other optimal coupling, so ambiguity does not significantly affect the accuracy in this example (assuming perfect predictability). About ambiguity, see Section \ref{sec:limitations}.} 

The predictability assumption allows us to find a predictor with little information. This is especially useful if we do not have the resources to collect paired data, and have to rely on separate datasets collected by others. We can also observe that the MCEC predictor works only when there is a class imbalance (the classes have different proportions). If the classes are evenly distributed (e.g., $50\%$ of the devices have the bug), then the classes become indistinguishable when only the marginal distribution of classes is known.

Another usage of MCEC is that (\ref{eq:mcec}) gives a lower bound (which we call the \emph{MCEC lower bound}) for the average cross entropy loss of \emph{any predictor} and \emph{any coupling} between the unpaired data. Unlike the MCEC predictor which relies on the predictability assumption, the MCEC lower bound applies without any assumption (other than that the marginals are accurate). It is useful in the following scenarios:
\begin{itemize}
\item If we have access to unpaired data from previous surveys, and would decide whether to collect paired data in order to build a predictor, then (\ref{eq:mcec}) gives a lower bound for the cross entropy loss of any predictor for any coupling between the unpaired data. If this lower bound is too high, we should not spend effort to collect paired data to build a predictor which is bound to fail. 
\item Even if we have paired data, the sample size may be insufficient to estimate the (high-dimensional) joint distribution or to learn an accurate predictor, but may be sufficient to estimate the (low-dimensional) marginal distributions. The MCEC lower bound (\ref{eq:mcec}) would serve as a guideline which lets us know the room for improvement of our current predictor if we collect more data, since it cannot improve beyond the lower bound. 
\item If we have a large number of features, and would perform feature selection to choose a small subset that is sufficient to give a good predictor, then (\ref{eq:mcec}) can rule out some subsets that are bound to fail.
\end{itemize}

The contributions of this paper are as follows:
\begin{itemize}
\item We propose the $\varLambda$\emph{-MCEC algorithm} based on linear programming that approximately finds the MCEC predictor (\ref{eq:mcec}), giving a coupling with $H(X|Y_{1},\ldots,Y_{m})$ within $O(\log\log N)$ from the optimum (where $N:=|\mathrm{supp}(P)|\prod_{i}|\mathrm{supp}(Q_{i})|$ is the dimension of the coupling), with a time and space complexity that is polynomial in $N$.\footnote{The dimension of the coupling $N$ is exponential in $m$. Given that the goal is to find a coupling, every algorithm that outputs a coupling as a full probability table would require a running time exponential in $m$.} The $O(\log\log N)$ gap can be significantly smaller than $H(P)=O(\log|\mathrm{supp}(P)|)$. Note that solving (\ref{eq:mcec}) or (\ref{eq:imec}) exactly is an NP-hard problem, due to minimum entropy coupling being NP-hard \cite{vidyasagar2012metric,kovavcevic2015entropy}, so we have to consider approximation algorithms.
\item Using synthetic data generated from the Dirichlet distribution, we show that the MCEC and IMEC predictors significantly outperform the uninformed predictor baseline (which only uses the class marginal $P$) in terms of macro F1 \cite{tsoumakas2010random,opitz2024closer},\footnote{We use macro F1 since it is a popular metric that is robust to class imbalance.} as long as $H(X|Y_{1},\ldots,Y_{m})$ is not too large. This shows that knowing the feature marginals $Q_{1},\ldots,Q_{m}$ is useful even without paired data. 
\item From experiments on real-world datasets, we observe that the MCEC and IMEC predictors can outperform the uninformed predictor when $H(X|Y_{1},\ldots,Y_{m})$ is small, but may underperform when $H(X|Y_{1},\ldots,Y_{m})$ is not small.
\item We study various properties of $H^{*}(P\Vert Q_1, \ldots, Q_m)$, such as the conditions under which \emph{perfect prediction} is possible, i.e., $H^{*}(P\Vert Q_1, \ldots, Q_m)=0$. We characterize $H^{*}$ when the distributions $P,Q_i$ are binary.
 \item We give conditions under which \emph{perfect prediction} is possible, i.e., $H^{*}(P\Vert Q_1, \ldots, Q_m)=0$ (there exists a coupling with $H(X|Y_{1},\ldots,Y_{m})=0$). These conditions are useful for finding out whether it is possible to have a perfect predictor given a set of marginals. We show that perfect prediction is possible if $Q_{1},\ldots,Q_{m}$ are ``rich enough''. In comparison, if we also impose the independent features assumption, perfect predictability may not be possible even for large $m$, which shows the importance of coupled $Y_{1},\ldots,Y_{m}$ for perfect predictability. 
\end{itemize}

\noindent\textbf{AI Use Disclosure:}
We derived Theorem~\ref{them:binary_hetero} in collaboration with an AI model. We also use the model for identifying gaps in other proofs and arguments in the draft (which are then corrected by humans), and also for identifying relevant literature.
We take full responsibility for correctness, presentation, and communication of the results.

\subsection{Previous Works}\label{sec:previous}
When $m=1$, we minimize $H(X|Y_{1})$, or equivalently, minimize $H(X,Y_{1})=H(X|Y_{1})+H(Q_{1})$. This problem is known as \emph{minimum entropy coupling} \cite{vidyasagar2012metric,painsky2013memoryless,kovavcevic2015entropy,kocaoglu2017entropic,kocaoglu2017greedy,cicalese2019minimum}, which has been 
extensively studied in information theory
\cite{EntropyCoupling_two_Sokota,shkel2023information,ebrahimi2024minimum,yu2018asymptotic,rossi2019greedy,kovavcevic2015entropy}
and underlies applications in causal inference
\cite{kocaoglu2017entropic,kocaoglu2017greedy,Javidian:21}, concept erasure
\cite{chowdhury2025fundamental}, steganography
\cite{sokota2023perfectly}, dimensionality reduction
\cite{cicalese2019minimum}, communication \cite{Sokota_Communication}, and
random number generation \cite{li2021efficient}.
The \emph{greedy minimum entropy coupling algorithm} \cite{kocaoglu2017greedy} is an efficient algorithm for finding a low-entropy coupling, which is briefly described in Appendix \ref{app:greedy}. The relation between entropic causal inference~\cite{kocaoglu2017entropic,kocaoglu2017greedy} and the predictability assumption is also discussed in Appendix \ref{app:greedy}.

A relevant prior work \cite{bounoua2025learning} utilizes minimum entropy coupling to match unpaired data. 
The novelty of MCEC is that we consider a multi-marginal coupling between $P,Q_1,\ldots,Q_m$ where we minimize the conditional entropy instead of the joint entropy, since our goal is to build a predictor instead of a matching.
%Although MCEC (\ref{eq:mcec}) also produces a matching, the goal is not to produce a matching, but to build a predictor. 
Also, unlike \cite{bounoua2025learning} which concerns continuous variables with parametrized conditional distribution (e.g., via denoising diffusion probabilistic model \cite{ho2020denoising}), we focus on categorical variables, so no structural assumptions on their values can be made.

% image-to-image translation \cite{zhu2017unpaired,liu2017unsupervised,kim2017learning,almahairi2018augmented} ???

The training of binary classifier from unlabeled data was studied in \cite{lu2018on,lu2020mitigating}, which corresponds to the setting $X \in \{\pm 1\}$ and $m=1$. Learning is possible if we have two sets of unlabeled data for $Y$ under different prior distributions of $X$, i.e., we have a dataset of $Y$ given $P_X(1)=\theta$, and another dataset of $Y$ given $P_X(1)=\theta' \neq \theta$ \cite{lu2018on}. This is a weakly-supervised learning setting since whether a data point is in the first or second dataset is related to $X$ through the priors. In comparison, in the setting in this paper, we would only have one dataset for $Y$ and another for $X$ (which allows us to learn the prior). This is, in a certain sense, an unsupervised learning setting since a data point for $Y$ is completely unrelated to $X$.

In learning from label proportions \cite{quadrianto2008estimating,zhang2022learning}, the class label of each data point is not known, but only the label proportions of bags of data points are known, so the data points are paired to the labels in a weak sense. In this paper, the pairing is completely absent.

\section{Minimum Conditional Entropy Coupling\label{sec:MCEC_prop}}
In this paper, we assume $P\in \Delta(\calX)$ (the probability simplex over a finite set $\calX$), $Q_i\in\Delta(\calY_i)$ ($\calY_i$ are finite sets).
We first study the properties of $H^{*}(P\Vert (Q_i))$ defined in \eqref{eq:mcec}. To this goal, let $\calC(Q_1, \dots, Q_m)$ be the collection of all couplings of $(Q_1, \dots, Q_m)$. 
The proofs of the results in this section are in Appendix \ref{app:proofs}.

\begin{proposition}\label{prop:properties}
    We have
    \begin{enumerate}
    \item $H^*(P\Vert Q_1,\ldots,Q_m)\le H(P)$.
    \item Monotonicity: We have  
    $$H^*(P\Vert Q_1,\ldots,Q_m,Q_{m+1}) \leq H^*(P\Vert Q_1,\ldots,Q_m).$$ 
    In fact, the monotonicity is strict unless  $H^*(P\Vert Q_1,\ldots,Q_m)=0$ or $H(Q_{m+1})=0$.
    \item %Reduction to a pairwise coupling: 
    %\begin{align*}
        $ H^*(P\Vert Q_1,\ldots,Q_m) =\displaystyle\min_{R\in\mathcal C }\min_{X\sim P, S\sim R} H(X| S)$,
    %\end{align*} 
    where $\calC$ denotes $\calC(Q_1,\ldots,Q_m)$.
    \item  
        $H^*(P\Vert Q_1,\ldots,Q_m)=0$ if and only if  there exists  $R\in\calC(Q_1,\ldots,Q_m)$ and a deterministic map $f$ such that the pushforward measure is $f_\#R=P.$ %\SA{Somewhat trivial.}
\item If $Z=g(X)$ and $g_\#P=P'$, then
$$ H^*(P'\Vert Q_1,\ldots,Q_m)\le  H^*(P\Vert Q_1,\ldots,Q_m).$$
\item Subadditivity:
\begin{align*}
    & H^*(P\otimes P'\Vert Q_1\otimes Q'_1,\ldots,Q_m\otimes Q'_m)\leq H^*(P\Vert Q_1,\ldots,Q_m)+H^*(P'\Vert Q'_1,\ldots,Q'_m).
\end{align*} 
\item Lower bound: 
\begin{align}
    H^*(P\Vert Q_1,\dots,Q_m) \ge  \max \Big\{ \Big [H(P)-\sum_{i=1}^m H(Q_i)\Big]_+, -\log \sum_{j=1}^{K} P^{\downarrow}(j)\Big\}, \label{eq:lower_bound}
\end{align}
where $K = \prod_{i=1}^m |\supp(Q_i)|$ and 
$P^{\downarrow}(j)$ is the $j$-th largest entry of $P$. 
 \end{enumerate}
\end{proposition}

Computing $H^{*}(P\Vert (Q_i)_i)$ exactly is an NP-hard problem since it becomes minimum entropy coupling when $m=1$, which was proved to be NP-hard in \cite{vidyasagar2012metric,kovavcevic2015entropy}. In this section, we focus on special cases where $H^*$ can be computed, and bounds on $H^*$.
Algorithms for estimating $H^*$ will be discussed in Section~\ref{sec:approx}.

Perfect prediction is possible in the best case, i.e., there is a coupling $(X,Y_1,\ldots,Y_m)$ where $X$ is a function of $(Y_1,\ldots,Y_m)$, if $H^{*}(P\Vert (Q_i)_i)=0$. By the lower bound in Proposition~\ref{prop:properties}, a necessary condition for perfect prediction is $\sum_{i=1}^m H(Q_i) \ge H(P)$. We prove the following sufficient condition that perfect prediction is possible if $(Q_i)$ contains two copies of a uniform distribution with support as large as the support of $P$. 
\begin{proposition}\label{prop:uniform}
If $r \ge |\mathrm{supp}(P)|$, then $H^*(P\| \Unif([r]), \Unif([r]))=0$, where $\Unif([r])$ is the uniform distribution over $[r]=\{1,\ldots,r\}$.
\end{proposition}

We then characterize $H^*$ when all the variables $X,Y_1,\ldots,Y_m$ are binary.  
%The proof is in Appendix \ref{app:proofs}. \SA{You already mentioned in the line above Prop 1 that all proofs are in Appendix B, so we don't have to keep saying it here. }
\begin{theorem}\label{them:binary_hetero}
Let $P=\Ber(p)$ and $Q_i=\Ber(q_i)$ be Bernoulli distributions with $p,q_i\in[0,1/2]$. Define
$t\coloneqq\sum_{i=1}^m q_i$ and $q^*\coloneqq\max_{i\in[m]}q_i.$
Then, we have 
\[
H^*(P\Vert Q_1, \ldots, Q_m) \!= \!
\begin{cases}
(1-t)h_{\mathsf b}\!\big(\frac{p-t}{1-t}\big), & p>t,\\
(2q^* \! -t) \min\big\{h_{\mathsf b}\!\big(\frac{p}{2q^*-t}\big),h_{\mathsf b}\!\big(\frac{p-t+q^*}{2q^*-t}\big)\big\}, & t-q^*<p<2q^* \! - \! t,\\
0, & \text{otherwise},
\end{cases}
\]
where $h_{\mathsf b}(p) := -p \log p - (1-p) \log (1-p)$ is the binary entropy function. 
\end{theorem}
In the binary case, the minimum conditional entropy is governed not only by $t=\sum_i q_i$, but also by how evenly that mass is spread across the coordinates. If no coordinate dominates, $2q^*\le t$, then every $X\sim \Ber(p)$ with $p\le t$ can be realized by a deterministic functions of $Y_1, \dots, Y_m$, and hence $H^*(P\Vert Q_1, \ldots, Q_m)=0$. If one coordinate dominates, $2q^*>t$, there is an unavoidable excess mass $2q^*-t$ attached to that coordinate. In this case, zero conditional entropy is not possible if $p$ falls in the range $t-q^*<p<2q^* - t$.

In general, we have the following upper bound which decays exponentially as $m$ increases. 
\begin{theorem}\label{thm:heterogeneous-exp} We have
\[
H^*(P\|Q_1,\ldots,Q_m)\le \min\Big\{1,(r-1)\prod_{i=1}^m\|Q_i\|_\infty\Big\}\log r,
\]
where $\Vert Q\Vert_{\infty} := \max_y Q(y)$, $r := |\mathrm{supp}(P)|$.
\end{theorem}
The above upper bound is obtained by a quantile coupling between $X\sim P$ and $Y^m\sim \otimes_{i=1}^m Q_i$ such that $X$ is recovered from $Y^m$ except on a mismatch event. The probability of this event is controlled by the largest atom of $Q_1\otimes\cdots\otimes Q_m$, which is $\prod_{i=1}^m\|Q_i\|_\infty$, yielding exponential decay whenever all $Q_i$ are non-deterministic. Since it implicitly assumes the observation $Y_1, \dots, Y_m$ are independent, this theorem can be viewed as a baseline for what can be achieved without exploiting dependence.
\subsection{Perfect prediction}
Theorem~\ref{thm:heterogeneous-exp} shows that $H^*(P\Vert Q_1,\ldots,Q_m)$ can decay exponentially in the number of features $m$. This naturally raises the question of when perfect prediction is possible; that is, when $H^*(P\Vert Q_1,\ldots,Q_m)=0.$ It is worth noting that perfect prediction is not possible in general under the independent features assumption. For example, let $P=\Ber(1/3)$ and let $Q_i=Q=\Ber(1/2)$ for all $i$. Under the independent feature assumption, $Y^m\sim Q^{\otimes m}$ is uniform on $\{0,1\}^m$. If $X=f(Y^m)$ for some deterministic function $f:\{0,1\}^m\to\{0,1\}$, then $\Pr(X=1)$ must be an integer multiple of $2^{-m}$. Since $1/3$ is not of this form for any finite $m$, such function $f$ does not exist. 

The following theorem provides a sufficient condition for perfect prediction in the homogeneous setting. For an integer $\ell$ and a distribution $Q$ on $[r]$ for some integer $r\geq 2$, we let $Q^{*\ell}$ denote the $\ell$-fold convolution of $Q$, that is, the distribution of $Y_1+Y_2+\dots+Y_\ell$ modulo $r$, where $Y_i\sim Q$.

\begin{theorem}\label{thm:equalZero}
Let $P,Q$ be finite discrete distributions on $[r]$ for an integer $r\ge2$, and assume that $P$ has full support. Suppose there exists an integer $\ell\ge1$ such that
\(Q^{*\ell}(x)>0,\) for all $x\in[r]$, where $Q^{*\ell}$ denotes the $\ell$-fold convolution of $Q$ over $[r]$. Then there exists an integer $m_0$ such that
$$H^*_m(P\|Q)=0,$$
for all $m\ge m_0$, where  $H^*_m(P\|Q) := H^*(P \| Q_1,\ldots,Q_m)$ with $Q_1=\cdots =Q_m = Q$.
\end{theorem}

The proof is a quantitative version of the uniform coupling structure given in the proof of Proposition~\ref{prop:uniform} and relies on three steps. First, the full-support condition on $Q^{*\ell}$ implies that long block sums of i.i.d.\ $Q$-samples modulo $r$ converge exponentially fast to the uniform distribution. Second, this near-uniformity allows us to construct a finite collection of coupled block variables whose differences approximate the vertices of the probability simplex. Third, once the approximation is accurate enough, these perturbed vertices span the simplex, so any target distribution $P$ can be represented exactly as a convex combination of them. The second step is precisely where dependence changes the nature of the problem. The block sums themselves are formed from independent $Q$-samples, but the coupling across blocks is engineered so that the resulting observation vector encodes $X$ deterministically. In contrast to the gradual entropy decay obtained from
independent features in Theorem~\ref{thm:heterogeneous-exp}, the dependent construction creates a finite-sample transition to zero entropy at a finite sample size $m=m_0$. Monotonicity of $m\mapsto H^*_m(P\|Q)$ then extends the zero-entropy construction from block length $m_0$ to all sample sizes $m\geq m_0$.

\iffalse

--------(TODO bounds for binary $Q_{i}$'s)--------
\begin{proposition}
We have
\[
H^{*}(P\Vert Q_{1},\ldots,Q_{m})\le\max\Big\{\Big\lceil H_{0}(P)-\sum_{i}H_{\infty}(Q_{i})\Big\rceil,\,0\Big\}+1.
\]
\end{proposition}
\smallskip{}

\smallskip{}

The following bound follows directly from the result in \cite{cicalese2019minimum} that there exists $R\in\mathcal{C}(P,Q)$ with $H(R)\le H(P\wedge Q)+1$ (also see \cite{rossi2019greedy,li2021efficient,compton2022tighter}).
\begin{theorem}
\label{thm:min_entropy}We have
\begin{align*}
H^{*}(P\Vert Q) & \le H(P\wedge Q)-H(Q)+1\\
 & \le\max\{H_{0}(P),H_{0}(Q)\}-H(Q)+1.
\end{align*}
Hence,
\[
H_{m}^{*}(P\Vert Q)\le\max\{H_{0}(P),mH_{0}(Q)\}-mH(Q)+1.
\]
\end{theorem}
\smallskip{}

Proposition ??? shows the lower bound $H_{m}^{*}(P\Vert Q)\ge H(P)-mH(Q)$. The following corollary shows that this lower bound is tight within $1$ bit for $m\le\lfloor H_{\infty}(P)/H_{0}(Q)\rfloor$.

\smallskip{}

\begin{corollary}
We have
\[
H_{m}^{*}(P\Vert Q)\le H(P)-\min\Big\{ m,\Big\lfloor\frac{H_{\infty}(P)}{H_{0}(Q)}\Big\rfloor\Big\} H(Q)+1.
\]
\end{corollary}
\begin{proof}
Let $m':=\min\{m,\lfloor H_{\infty}(P)/H_{0}(Q)\rfloor\}$. We have $H_{\infty}(P)\ge m'H_{0}(Q)=H_{0}(Q^{\otimes m'})$, and hence $P\preceq Q^{\otimes m'}$, and $P\wedge Q^{\otimes m'}$ is a rearrangement of $P$. By Theorem \ref{thm:min_entropy}, $H^{*}(P\Vert Q^{\otimes m'})\le H(P\wedge Q^{\otimes m'})-H(Q^{\otimes m'})+1=H(P)-m'H(Q)+1$.
\end{proof}
\smallskip{}

\fi

\section{Approximate MCEC\label{sec:approx}}

\subsection{The $\varLambda$-MCEC Linear Program}

Since computing the optimal coupling is NP-hard, we propose an approximation algorithm based on linear programming instead. We require two technical tools. First, $P$ is \emph{majorized} by $Q$ \cite{marshall1979inequalities_old}, denoted as $P\preceq Q$, if $\sum_{i=1}^{k}P^{\downarrow}(i)\le\sum_{i=1}^{k}Q^{\downarrow}(i)$ for every $k\ge1$, where $P^{\downarrow}(i)$ is the $i$-th largest entry of $P$ ($P^{\downarrow}(i)=0$ if $P$ has fewer than $i$ entries). Entropy is Schur-concave, i.e., $H(P)\ge H(Q)$ if $P\preceq Q$ \cite{marshall1979inequalities_old}. The space of discrete distributions is a lattice with partial order $\preceq$, where the greatest lower bound of $P,Q$, denoted as $P\wedge Q$, is a distribution over $\mathbb{N}$ defined as \cite{cicalese2002supermodularity,marshall1979inequalities_old} 
% \begin{align*}
% (P\wedge Q)(k) & :=\min\Big\{\sum_{i=1}^{k}P^{\downarrow}(i),\sum_{i=1}^{k}Q^{\downarrow}(i)\Big\} \\
% & \;\;\;\;\;\; -\min\Big\{\sum_{i=1}^{k-1}P^{\downarrow}(i),\sum_{i=1}^{k-1}Q^{\downarrow}(i)\Big\}.
% \end{align*}
\begin{align*}
(P\wedge Q)(i)  & :=\min\{P^{\downarrow}(\le i),Q^{\downarrow}( \le i)\} -\min\{P^{\downarrow}(\le i-1),Q^{\downarrow}(\le i-1)\},
\end{align*}
where $P^{\downarrow}(\le i) := \sum_{j=1}^{i} P^{\downarrow}(j)$.
The second tool we use is the \emph{discrete layered entropy} studied in \cite{shannon1951prediction,li2026discrete}:
\[
\varLambda(P):=\textstyle\sum_{i=1}^{\infty}P^{\downarrow}(i)(i\log i-(i-1)\log(i-1)),
\]
which is a concave, Schur-concave and piecewise linear function satisfying $\varLambda(P)\le H(P)$. For a random variable $X$, write $\varLambda(X):=\varLambda(P_{X})$. 
Define
\[
H_{\wedge}^{*}(P\Vert(Q_{i})_{i}):=\inf_{R\in\mathcal{C}(Q_{1},\ldots,Q_{m})}\big(H(P\wedge R)-H(R)\big),
\]
\[
\varLambda_{\wedge}^{*}(P\Vert(Q_{i})_{i}):=\inf_{R\in\mathcal{C}(Q_{1},\ldots,Q_{m})}\big(\varLambda(P\wedge R)-\varLambda(R)\big).
\]
For brevity, we write $H^{*}=H^{*}(P\Vert Q_1, \ldots, Q_m)$, and similar for $H_{\wedge}^{*},\varLambda_{\wedge}^{*}$.  By \cite{cicalese2019minimum} (also see \cite{rossi2019greedy,li2021efficient,compton2022tighter}),
\[
H_{\wedge}^{*}\le H^{*}\le H_{\wedge}^{*}+\log2.
\]
Given a coupling $R\in\mathcal{C}(Q_{1},\ldots,Q_{m})$ attaining the minimum in $H_{\wedge}^{*}$, we can apply the greedy minimum entropy coupling algorithm \cite{kocaoglu2017greedy} on $P,R$ to find a coupling $P_{X,\mathbf{Y}}\in\mathcal{C}(P,R)$ (write $\mathbf{Y}=(Y_1,\ldots,Y_m)$) with $H(P_{X,\mathbf{Y}})\le H(P\wedge R)+\log2$ \cite{rossi2019greedy}, so $H(X|\mathbf{Y})\le H(P\wedge R)-H(R)+\log2=H_{\wedge}^{*}+\log2$. If we can find the optimal coupling for $H_{\wedge}^{*}$, we can convert it to an almost-optimal coupling for $H^{*}$ within $\log2$ from the optimum. However, $H(P\wedge R)-H(R)$ appears to be as difficult to minimize as $H(X|\mathbf{Y})$. Hence, we instead use $\varLambda_{\wedge}^{*}$ as an efficiently computable approximation of $H^{*}$. We first show that $\varLambda_{\wedge}^{*}$ is an approximate lower bound of $H^*$. The proof is in Appendix~\ref{app:MCEC_proofs}.
\begin{theorem}
\label{thm:Lambda_le_Hstar}For every $\eta>0$,
\begin{equation}
\varLambda_{\wedge}^{*}\le H^{*}+\log(1+H^{*}/(e\eta))+\eta. \label{eq:Lambda_MCEC_Hstar}
\end{equation}
% The optimal choice of $\eta$ which gives the best bound is $\eta=\frac{1}{2e}(\sqrt{(H^{*})^{2}+4eH^{*}\log e}-H^{*})$. %\footnote{Although we use base-$e$ logarithms so $\log e=1$, we keep $\log e$ in expressions so they will still be valid if $\log$ and entropy are base-$2$ instead of base-$e$.} 
Hence, we have 
\begin{equation}
H^{*}\ge \max\{ g^{-1}(\varLambda_{\wedge}^{*}),\, \mathrm{LB} \},\label{eq:Lambda_MCEC_lb}
\end{equation}
where $g^{-1}$ is the inverse function of the increasing function $g(t):=t+\log(1+t/(e\eta))+\eta$ with $\eta=\frac{1}{2e}(\sqrt{t^{2}+4et\log e}-t)$ for $t>0$,\footnote{$g^{-1}$ can be found by binary search or Newton method.} and $\mathrm{LB}$ is the lower bound in \eqref{eq:lower_bound}. We call \eqref{eq:Lambda_MCEC_lb} the $\varLambda$\emph{-MCEC lower bound} for $H^{*}$.  In particular, if $\varLambda_{\wedge}^{*}>0$, then $H^{*}>0$, and perfect prediction is impossible.
\end{theorem}

As a result, $\varLambda_{\wedge}^{*} \approx H^{*}$. The proof is in Appendix~\ref{app:MCEC_proofs}.
\begin{corollary}
\label{cor:Lambda_approx_Hstar} 
\[
|\varLambda_{\wedge}^{*}-H^{*}|\le\log(1+e^{-1}\log N)+\log(2e),
\]
where $N:=|\mathrm{supp}(P)|\prod_{i=1}^{m}|\mathrm{supp}(Q_{i})|$.
\end{corollary}

Next, we study the computation of $\varLambda_{\wedge}^{*}$.

\begin{theorem}
\label{thm:linear}For finite discrete $P,Q_{1},\ldots,Q_{m}$, $\varLambda_{\wedge}^{*}(P\Vert(Q_{i})_{i})$ is a piecewise linear convex function of $(Q_{i})_{i}$ for any fixed $P$. 
It can be computed via a linear program with $O(N)$ variables and constraints.
\end{theorem}

\begin{proof}
We first prove that $\varLambda_{\wedge}^{*}(P\Vert R) = \varLambda(P \wedge R) - \varLambda(R)$ is convex in $R$. Let $n_{X}:=|\mathrm{supp}(P)|$. Let $\theta_{i}:=i\log i-(i-1)\log(i-1)$, $\psi_{i}:=\theta_{i+1}-\theta_{i}>0$. Write $P^{\downarrow}(\le i):=\sum_{j=1}^{i}P^{\downarrow}(j)$.
\begin{align}
 \varLambda_{\wedge}^{*}(P\Vert R) \, & =\,  \varLambda(P \wedge R) - \varLambda(R) \nonumber \\
 & =\sum_{i=1}^{\infty}\theta_{i}\big(\min\left\{ P^{\downarrow}(\le i),R^{\downarrow}(\le i)\right\}  -R^{\downarrow}(\le i) -\min\left\{ P^{\downarrow}(\le i-1),R^{\downarrow}(\le i-1)\right\} +R^{\downarrow}(\le i-1)\big)\nonumber \\
 % & =\sum_{i=1}^{\infty}\theta_{i}\big(\min\left\{ P^{\downarrow}(\le i)-R^{\downarrow}(\le i),0\right\} \nonumber \\
 % & \quad -\min\left\{ P^{\downarrow}(\le i-1)-R^{\downarrow}(\le i-1),0\right\} \big)\nonumber \\
 & =\sum_{i=1}^{n_{X}-1}\psi_{i}\left(\max\left\{ R^{\downarrow}(\le i)-P^{\downarrow}(\le i),\,0\right\} \right).\label{eq:Lambda_lp}
\end{align}
As $R^{\downarrow}(\le i)$ is convex in $R$, $\max\{R^{\downarrow}(\le i)-P^{\downarrow}(\le i),\,0\}$ and $\varLambda_{\wedge}^{*}(P\Vert R)$ are convex in $R$. Hence, $\varLambda_{\wedge}^{*}(P\Vert(Q_{i})_{i})=\min_{R\in\mathcal{C}(Q_{1},\ldots,Q_{m})}\varLambda_{\wedge}^{*}(P\Vert R)$ is a minimization problem with a convex piecewise-linear objective, and with a constraint $R\in\mathcal{C}(Q_{1},\ldots,Q_{m})$ that is linear in $(Q_{i})_{i}$. $\varLambda_{\wedge}^{*}(P\Vert(Q_{i})_{i})$ can be found via linear programming, and its optimal value is a piecewise linear convex function of $(Q_{i})_{i}$.\footnote{This is due to the fact that $\min_{\mathbf{x} \in \mathbb{R}^n: \mathbf{B}\mathbf{x} \le \mathbf{y}} f(\mathbf{x})$ is a piecewise linear convex function in $\mathbf{y} \in \mathbb{R}^m$, where $\mathbf{B} \in \mathbb{R}^{m \times n}$, and $f:\mathbb{R}^n \to \mathbb{R}$ is piecewise linear and convex.} Representing $R^{\downarrow}(\le i)$ using the technique in \cite{ogryczak2003minimizing}, we can find $\varLambda_{\wedge}^{*}(P\Vert(Q_{i})_{i})$ via the following linear program with variables $R,\tilde{R}_{i}(\mathbf{y}),t_{i},b_{i}$ ($1\le i\le n_{X}-1$, $\mathbf{y}\in\mathcal{M}:=\prod_{i=1}^{m}\mathrm{supp}(Q_{i})$):
\begin{align}
& \mathrm{minimize}\; \textstyle\sum_{i=1}^{n_{X}-1}\psi_{i}b_{i}\nonumber \\
& \mathrm{s.t.}\;\;  R\in\mathcal{C}(Q_{1},\ldots,Q_{m}),\,\tilde{R}_{i}(\mathbf{y})\ge0,\,t_{i}\ge0,\,b_{i}\ge0,\nonumber \\
 & \quad\quad b_{i}\ge it_{i}+ \textstyle\sum_{\mathbf{y}\in\mathcal{M}}\tilde{R}_{i}(\mathbf{y})-P^{\downarrow}(\le i),\nonumber \\
 & \quad\quad t_{i}+\tilde{R}_{i}(\mathbf{y})\ge R(\mathbf{y}),\label{eq:Lambda_LP_long}
\end{align}
for all $i\in\{1,\ldots,n_{X}-1\}$, $\mathbf{y}\in\mathcal{M}$.
\end{proof}

This linear program provides a method to compute $\varLambda_{\wedge}^{*}$ (the optimal value) and the $\varLambda$-MCEC lower bound \eqref{eq:Lambda_MCEC_lb}.

%The proof of Theorem \ref{thm:linear} reveals a linear programming algorithm for finding $\varLambda_{\wedge}^{*}$. The precise linear program will be given in (\ref{eq:Lambda_LP_long}).

\subsection{Refinement Algorithm}\label{subsec:refine}

The linear program \eqref{eq:Lambda_LP_long} computes a coupling $R\in\mathcal{C}(Q_{1},\ldots,Q_{m})$. If an overall coupling $\pi\in\mathcal{C}(P,Q_{1},\ldots,Q_{m})$ is needed, we apply the greedy minimum entropy coupling algorithm \cite{kocaoglu2017greedy} (see Appendix \ref{app:greedy}) on $P,R$ to find a coupling $\pi$ of $P$ and $R$. 
To further reduce the conditional entropy $H_{\pi}(X|Y_{1},\ldots,Y_{m})$ evaluated over the coupling $\pi$, we may apply the following refinement algorithm based on majorize-minimization \cite{lange2016mm}.

\begin{algorithm}[ht]
\caption{Refinement algorithm}
\label{alg:Hstar-refined}
\begin{algorithmic}[1]
\Require Initial coupling $\pi^{(0)}\in\mathcal{C}(P,Q_{1},\ldots,Q_{m})$ and number of iterations $T$
\For{$t=0,\ldots,T-1$}
\State $\Gamma_{t}:=\mathrm{supp}(\pi^{(t)})$
\State $\pi^{(t+1)}\! :=\mathrm{argmin}_{\pi\in\mathcal{C}(P,Q_{1},\ldots,Q_{m}):\mathrm{supp}(\pi)\subseteq\Gamma_{t}}L_{t}(\pi)$ computed via linear programming, with
\[
L_{t}(\pi):=\sum_{(x,\mathbf{y}) \in \Gamma_t}\pi(x,\mathbf{y})\log\frac{\sum_{x'}\pi^{(t)}(x',\mathbf{y})}{\pi^{(t)}(x,\mathbf{y})},
\]
where $\mathbf{y} \in \prod_i \mathrm{supp}(Q_i)$.
\EndFor
\State \Return $\pi^{(T)}$
\end{algorithmic}
\end{algorithm}

To show that $H_{\pi^{(t)}}(X|\mathbf{Y})$ (the conditional entropy evaluated over $\pi^{(t)}$) is non-increasing, note that $H_{\pi}(X|\mathbf{Y})$ is a concave function of $\pi$,\footnote{This is due to $H(X|\mathbf{Y})\ge H(X|\mathbf{Y},U)=\sum_{u}p_{U}(u)H(X|\mathbf{Y},U=u)$ for any random variable $U$.} and $L_{t}(\pi)=H_{\pi^{(t)}}(X|\mathbf{Y})+\langle\nabla_{\pi}H_{\pi}(X|\mathbf{Y})|_{\pi=\pi^{(t)}},\,\pi-\pi^{(t)}\rangle$ is the linear approximation of $H_{\pi}(X|\mathbf{Y})$ at $\pi=\pi^{(t)}$. Hence, each update minimizes a tight linear upper bound of $H_{\pi}(X|\mathbf{Y})$, and we have
\begin{align*}
H_{\pi^{(t+1)}}(X|\mathbf{Y}) &\le L_{t}(\pi^{(t+1)}) \le L_{t}(\pi^{(t)})=H_{\pi^{(t)}}(X|\mathbf{Y}).
\end{align*}
Hence, the refinement algorithm always reduces the conditional entropy. $H_{\pi^{(t)}}(X|\mathbf{Y})$ converges by monotonicity, but may not converge to $H^*(P \Vert (Q_i)_i)$.

\subsection{The $\varLambda$-MCEC and IMEC Algorithms \label{subsec:algo}}

We now describe the \emph{(refined)} $\varLambda$\emph{-MCEC algorithm}, an algorithm for finding an almost-optimal coupling:
\begin{enumerate}
\item (Layered entropy step) Use the linear program \eqref{eq:Lambda_LP_long} to find $R\in\mathcal{C}(Q_{1},\ldots,Q_{m})$ attaining the minimum of $\varLambda_{\wedge}^{*}(P\Vert R)$ in (\ref{eq:Lambda_lp}).
% \[
% \varLambda_{\wedge}^{*}(P\Vert R)=\sum_{i=1}^{n_{X}-1}\psi_{i}\big(\max\{ R^{\downarrow}(\le i)-P^{\downarrow}(\le i),0\} \big).
% \]
%where $\psi_{i}:=(i+1)\log(i+1)-2i\log i+(i-1)\log(i-1)$, $n_{X}=|\mathrm{supp}(P)|$, and $R^{\downarrow}(\le i)$ is the sum of the largest $i$ entries of $R$. 
If there are multiple minimizers, select the $R$ which maximizes $\min_{\mathbf{y}\in\mathcal{M}}R(\mathbf{y})$ (so the resultant $R$ would have larger entropy).\footnote{In practice, we simply minimize $\sum_{i=1}^{n_{X}-1}\psi_{i}b_{i}-\epsilon\min_{\mathbf{y}\in\mathcal{M}}R(\mathbf{y})$ in the linear program \eqref{eq:Lambda_LP_long}, where $\epsilon>0$ is small (we take $\epsilon=0.001$ in the experiments). If the goal is only to compute the $\varLambda$-MCEC lower bound, take $\epsilon=0$.} %If the goal is only to compute the $\varLambda$-MCEC lower bound, we can stop and output $g^{-1}(\varLambda_{\wedge}^{*})$ in (\ref{eq:Lambda_MCEC_lb}), where $\varLambda_{\wedge}^{*}$ is the optimal value of the linear program.
\item (Greedy coupling step) Use the greedy minimum entropy coupling algorithm \cite{kocaoglu2017greedy} (see Appendix \ref{app:greedy}) on $P$ and $R$ to find a coupling $P_{X,\mathbf{Y}}\in\mathcal{C}(P,R)$ with $H(P_{X,\mathbf{Y}})\le H(P\wedge R)+\log2$ \cite{rossi2019greedy}.
\item (Refinement step, optional) Use Algorithm \ref{alg:Hstar-refined} to refine the coupling $P_{X,\mathbf{Y}}$. If we use this step, we call this the \emph{refined} $\varLambda$\emph{-MCEC algorithm}. If we skip this step, we call this the $\varLambda$\emph{-MCEC algorithm}. 
\item (Predictor computation) The $\varLambda$-MCEC predictor is the conditional distribution $P_{X|\mathbf{Y}}(x|\mathbf{y})$ derived from the coupling. If $P_{\mathbf{Y}}(\mathbf{y})=0$, take $P_{X|\mathbf{Y}}(x|\mathbf{y})=P(x)$. If a hard-decision predictor is needed, we use $\hat{x}(\mathbf{y})=\mathrm{argmax}_{x}P_{X|\mathbf{Y}}(x|\mathbf{y})$.
\end{enumerate}
\smallskip{}

For the \emph{independent minimum entropy coupling (IMEC) algorithm}, we only run the greedy coupling step on $R=Q_{1}\otimes\cdots\otimes Q_{m}$, and the predictor computation step.

The (refined) $\varLambda$-MCEC algorithm gives an almost optimal coupling, as shown in the following result. The proof is in Appendix~\ref{app:MCEC_proofs}.

\smallskip{}

\begin{corollary}\label{cor:MCEC_almost_optimal}
The difference between the $H(X|Y_1,\ldots,Y_m)$ of the coupling given by the $\varLambda$-MCEC algorithm (or the refined $\varLambda$-MCEC algorithm) and $H^{*}(P\Vert Q_1, \ldots, Q_m)$ is at most
$2\log(1+e^{-1}\log N)+\log(2e^2)$,
where $N:=|\mathrm{supp}(P)|\prod_{i=1}^{m}|\mathrm{supp}(Q_{i})|$. 
\end{corollary}

\section{Experiments\label{sec:experiments}}

\subsection{Simulation Experiments on Predictors\label{subsec:sim}}

We evaluate the performance of the $\varLambda$-MCEC predictor on randomly generated joint distributions $P_{X,\mathbf{Y}}$ where $X\in[n_{X}]$ and $Y_{i}\in[n_{Y}]$ (recall that $\mathbf{Y}=(Y_i)_i$). To generate $P_{X,\mathbf{Y}}$, we first generate $P_{\mathbf{Y}}\in\Delta([n_{Y}]^{m})$ to be uniformly distributed over the simplex (so the $Y_i$'s may be dependent), and then generate the conditional distribution $P_{X|\mathbf{Y}}(\cdot|\mathbf{y})$ following the Dirichlet distribution with parameter $(\alpha,\ldots,\alpha) \in \mathbb{R}^{n_{X}}$ for each $\mathbf{y}$. We parametrize the Dirichlet distribution by the expected conditional entropy $h\in[0,\log n_{X}]$ instead of $\alpha$ (select $\alpha$ so that $\mathbb{E}[H(P_{X|\mathbf{Y}}(\cdot|\mathbf{y}))]=h$, expectation is over the Dirichlet random vector). 
$h$ is the expected value of $H(X|\mathbf{Y})$.

\begin{figure*}[ht]
\begin{centering}
\includegraphics[scale=0.255]{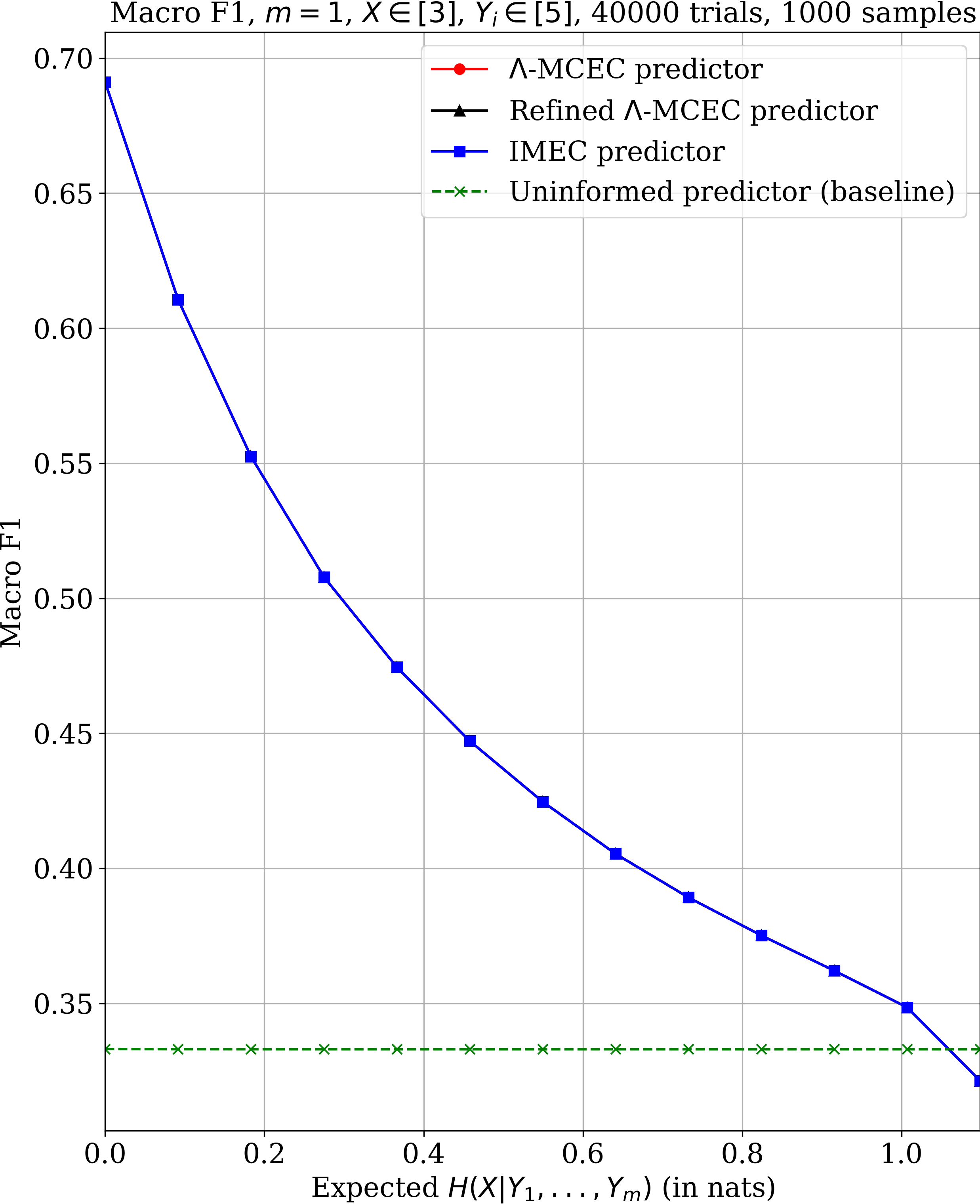} %$\;$
\includegraphics[scale=0.255]{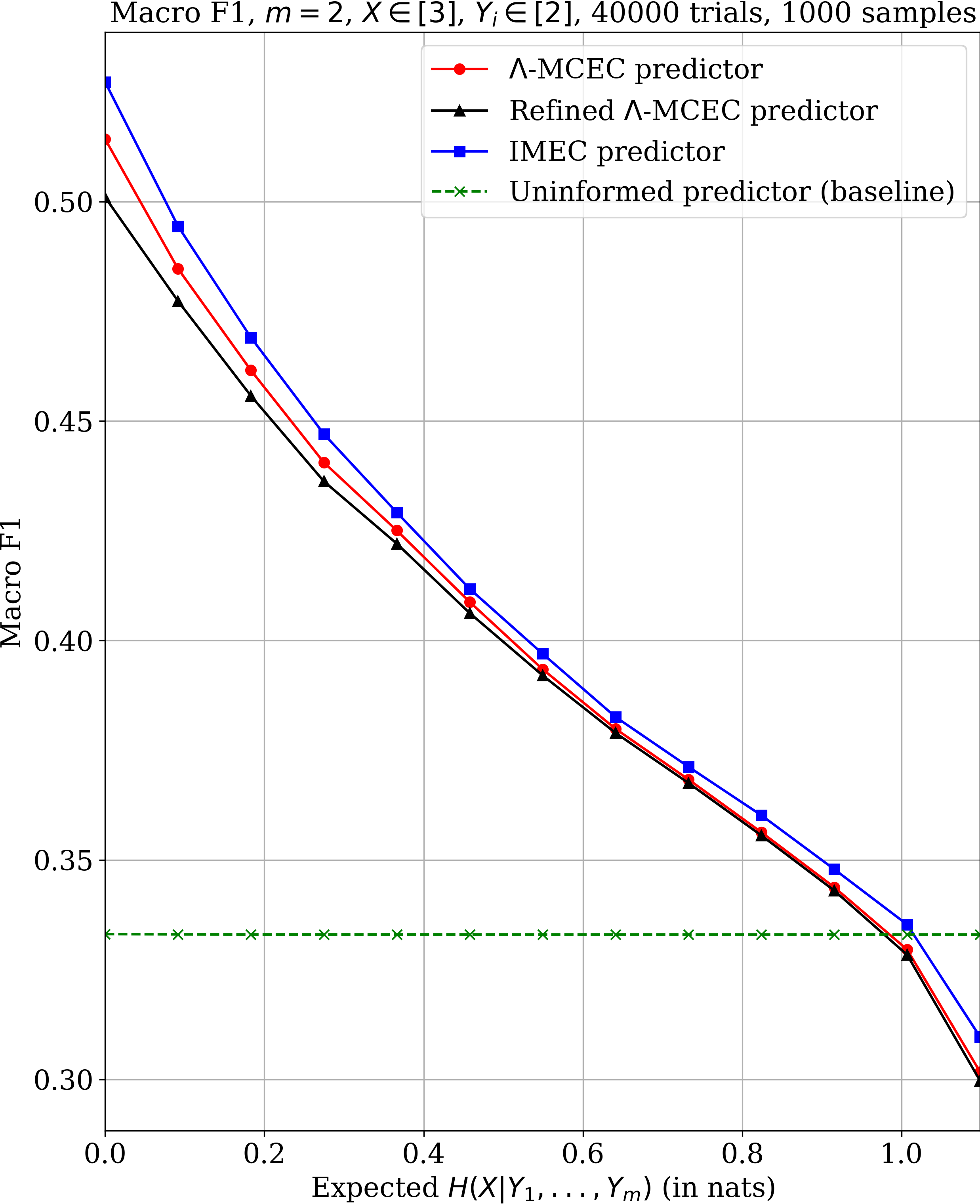} %$\;$
\includegraphics[scale=0.255]{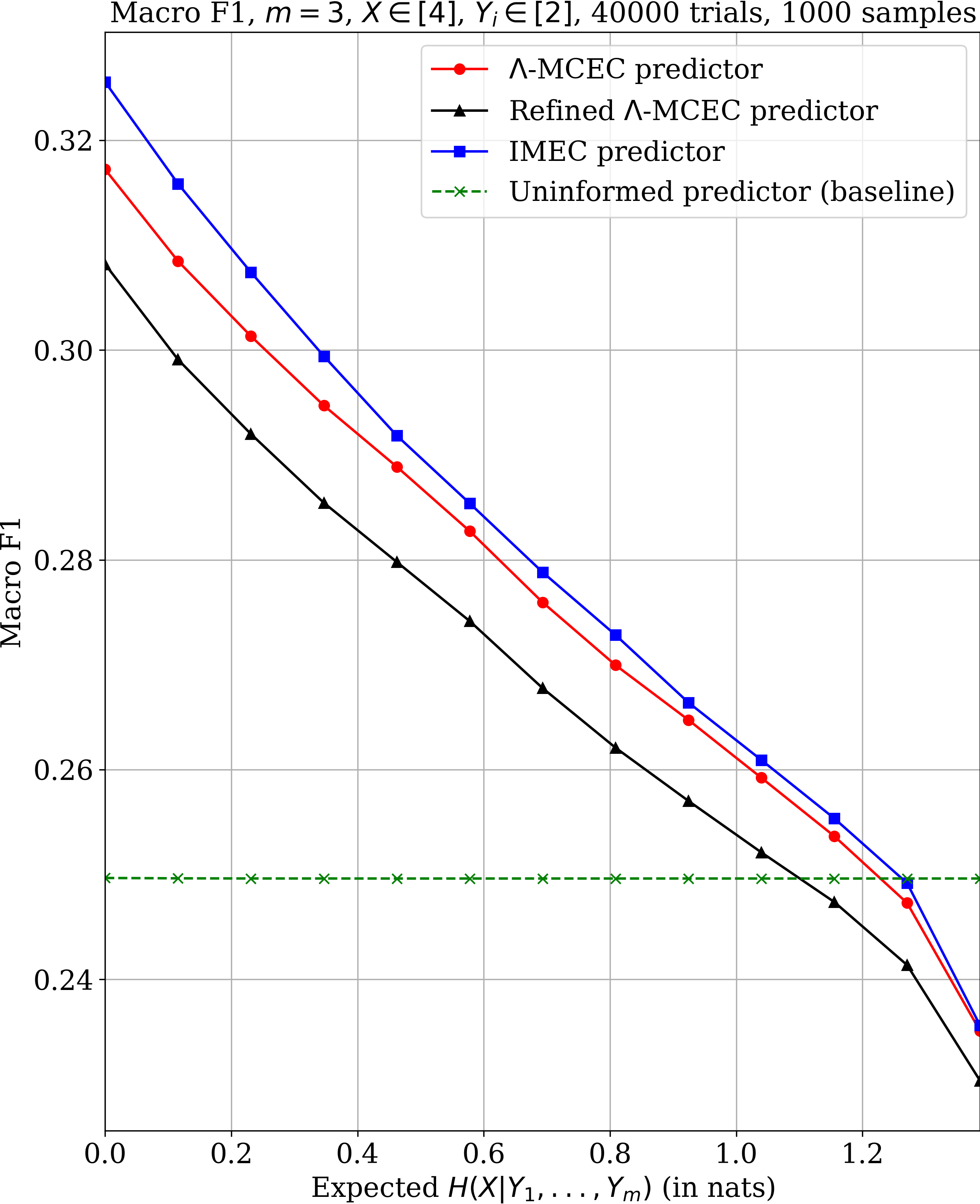}
\par\end{centering}
\caption{Comparison between $\varLambda$-MCEC, refined $\varLambda$-MCEC, IMEC and uninformed predictors, for $h\in[0,\log n_{X}]$ (expected value of $H(X|\mathbf{Y})$), and various choices of $m,n_{X},n_{Y}$ indicated in the titles of the subplots.}\label{fig:compare}
\end{figure*}

In Figure \ref{fig:compare}, we plot the macro F1 metric \cite{tsoumakas2010random,opitz2024closer} 
\[
\mathrm{F1}:=\frac{1}{n_{X}}\sum_{x\in\mathrm{supp}(P)}\frac{2\mathbb{P}(X=\hat{x}(\mathbf{Y})=x)}{\mathbb{P}(\hat{x}(\mathbf{Y})=x)+\mathbb{P}(X=x)},
\]
where $\hat{x}:\prod_{i}\mathrm{supp}(Q_{i})\to\mathrm{supp}(P)$ is a hard-decision predictor, and the probabilities are over the empirical distribution of the testing data.\footnote{Macro F1 lies in $[0,1]$, where a larger value indicates a better predictor. We use macro F1 since it is a popular metric that is robust to class imbalance. We do not use the cross entropy loss as an evaluation metric since it is infinite for hard-decision predictors that are not $100\%$ accurate. It unfairly penalizes hard decisions. Using cross entropy loss for training and another metric (e.g., macro F1) for evaluation is in line with the usual practice.} Each data point is an average of $40000$ trials. In each trial, we generate one distribution $P_{X,\mathbf{Y}}$ using the aforementioned procedure, and then generate $1000$ samples of $(X,\mathbf{Y})$ from that distribution for the training data, and another $1000$ samples for the testing data (for evaluating the macro F1).\footnote{\label{footnote:error_bar}We do not plot the error bars since the number of trials $40000$ is large enough that the standard error is at most $1/(2\sqrt{40000})=0.0025$, which is small.} We compare the following predictors:\footnote{\label{footnote:spec}Experiments were run on a computer with Intel Core i7-1065G7 CPU @1.3GHz, 16GB RAM. Each trial takes $\approx 0.13$ sec on average. The Python code is in the supplementary material (ZIP file).}
\begin{itemize}
\item $\varLambda$-MCEC predictor, computed by running the $\varLambda$-MCEC algorithm (without refinement) on the marginal distributions $P=\hat{P}_{X}$ (empirical distribution of $X$ in the training data), $Q_{i}=\hat{P}_{Y_{i}}$ (empirical distribution of $Y_{i}$) to obtain the predictor $\hat{P}_{X|\mathbf{Y}}$, which is then converted to a hard-decision predictor $\hat{x}(\mathbf{y})=\mathrm{argmax}_{x}\hat{P}_{X|\mathbf{Y}}(x|\mathbf{y})$.  
\item Refined $\varLambda$-MCEC predictor, with a refinement step to reduce the conditional entropy.
\item IMEC predictor.
\item Uninformed predictor, which outputs $\hat{x}$ randomly according to the empirical distribution $\hat{P}_{X}$, attaining $\mathrm{F1} \approx 1/n_{X}$ for large sample size. $1/n_{X}$ is the largest possible macro F1 for predictors without access to $\mathbf{Y}$.  This is the baseline predictor for learning with totally unpaired data, since there is no prior reason to believe that $\mathbf{Y}$ or their marginals are useful for predicting $X$ without information on their dependency.  %It would be interesting if $\varLambda$-MCEC outperforms this baseline. 
\end{itemize}

We discuss the findings in Figure \ref{fig:compare}.
% \begin{itemize}
% \item 
$\varLambda$-MCEC outperforms the uninformed predictor as long as the true $H(X|\mathbf{Y})$ is not large. It performs worse when $H(X|\mathbf{Y})$ is close to the theoretical maximum $\log n_X$. This is expected, since MCEC relies on the predictability assumption that $H(X|\mathbf{Y})$ is small. %, and would not be as performant as the baseline if the assumption is untrue.

% \item 
IMEC has a similar (and slightly better) performance compared to $\varLambda$-MCEC, even though the independent features assumption is not satisfied by the true distribution. This may be because IMEC reduces the number of parameters of the model (the coupling of $P,(Q_i)$) by fixing the joint distribution of $Y_1,\ldots,Y_m$. This phenomenon is similar to the naive Bayes classifier \cite{domingos1997optimality}, where an inaccurate assumption may lead to good performance due to a reduction of the number of parameters. $\varLambda$-MCEC might still be preferable since it relies on fewer assumptions.
%\item $\varLambda$-MCEC performs worse when $m$ is larger. Intuitively, as $m$ becomes larger, the dimension of the space of couplings between $P,Q_{1},\ldots,Q_{m}$ becomes larger, and it becomes more difficult to estimate the correct coupling given the low-dimensional marginal distributions. The problem becomes more underdetermined.

% \item 
Although refined $\varLambda$-MCEC attains a smaller conditional entropy, it performs slightly worse than $\varLambda$-MCEC. Intuitively, refined $\varLambda$-MCEC is ``overfitting'' $\hat{P}_{X|\mathbf{Y}}$ to make $X$ almost a function of $\mathbf{Y}$, which may not be true in reality. Refined $\varLambda$-MCEC is still useful when our goal is to obtain a lower bound of the performance of any predictor for any coupling between the marginal distributions (see the introduction and Appendix \ref{app:MCEC_bounds}).
% \end{itemize}

Refer to Appendix \ref{app:more} for comparisons against more predictors (e.g., MAP predictor and naive Bayes \cite{domingos1997optimality}).

\subsection{Experiments with Datasets\label{subsec:datasets}}

\begin{figure*}
\begin{centering}
\includegraphics[scale=0.185]{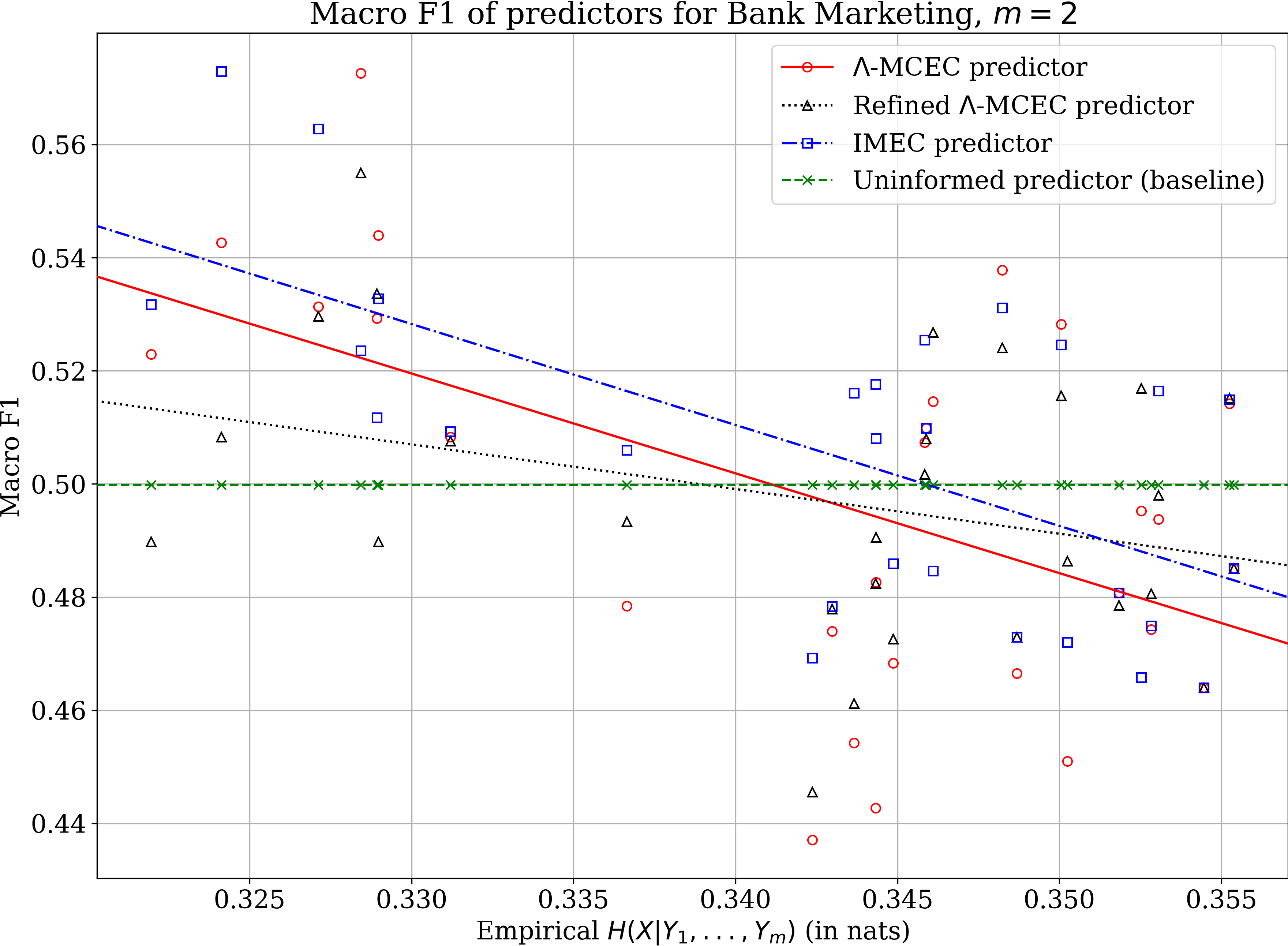}
\includegraphics[scale=0.185]{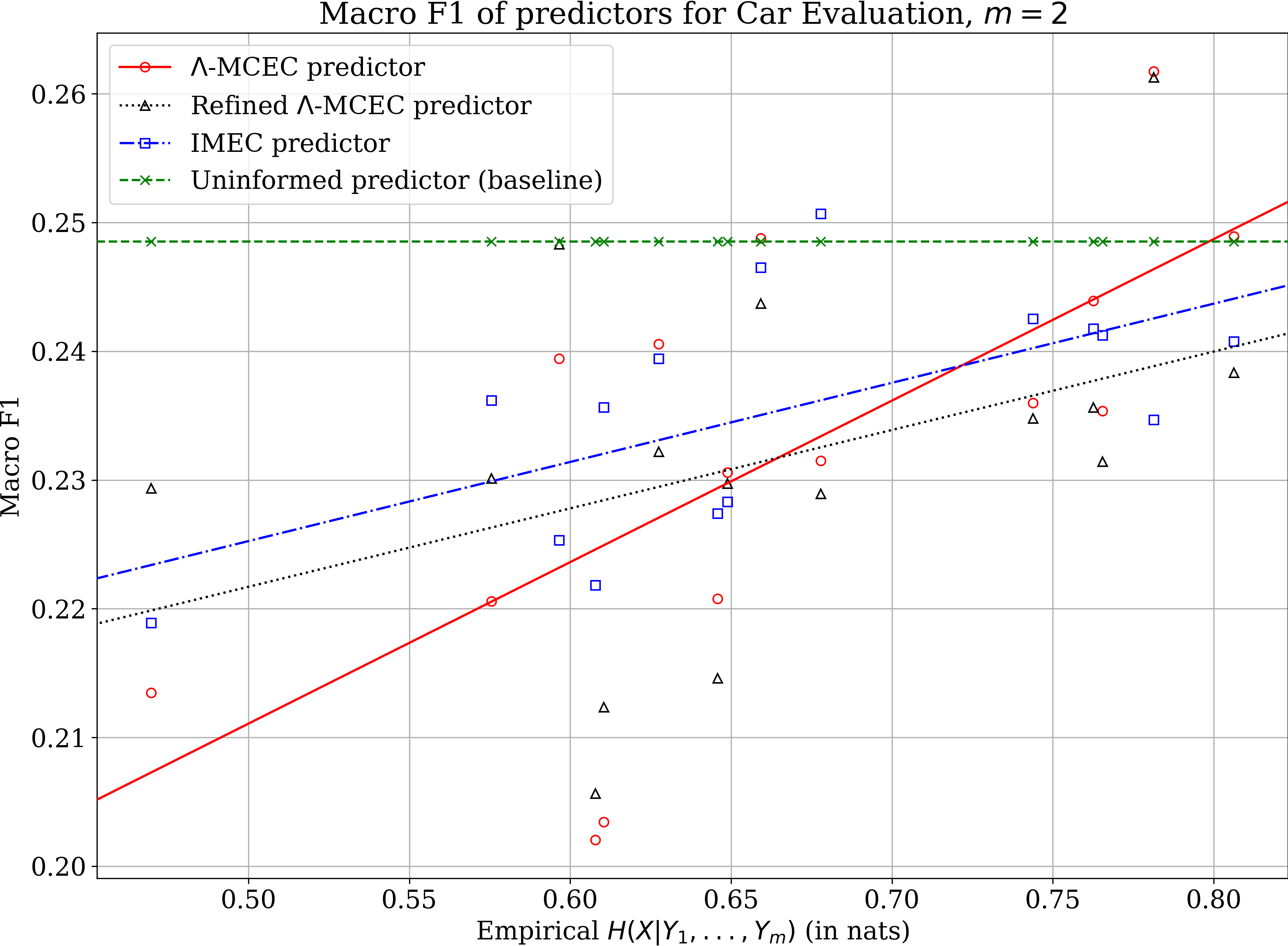}
\includegraphics[scale=0.185]{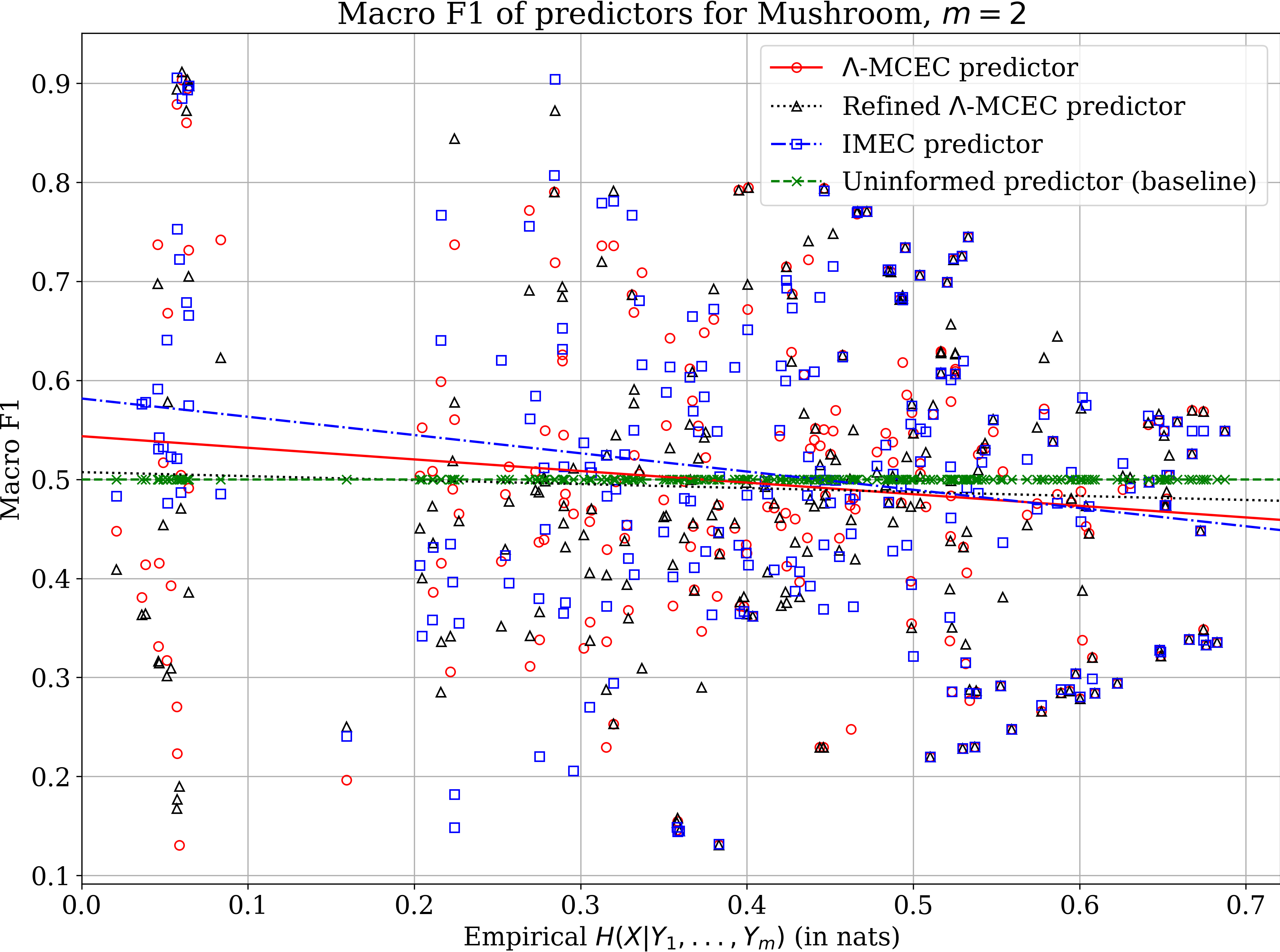}
\par\end{centering}
% \medskip{}

\begin{centering}
\includegraphics[scale=0.185]{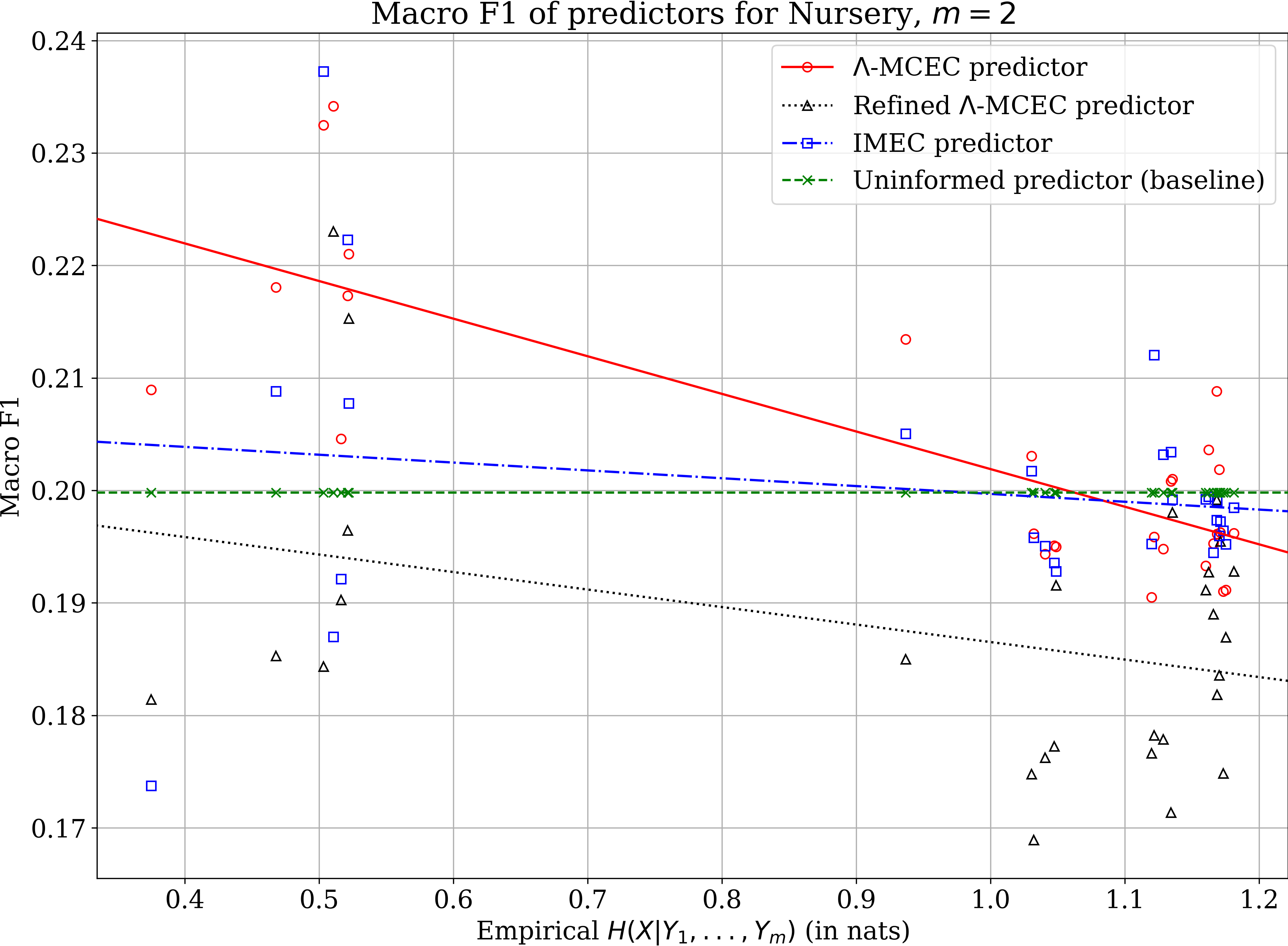}
\includegraphics[scale=0.185]{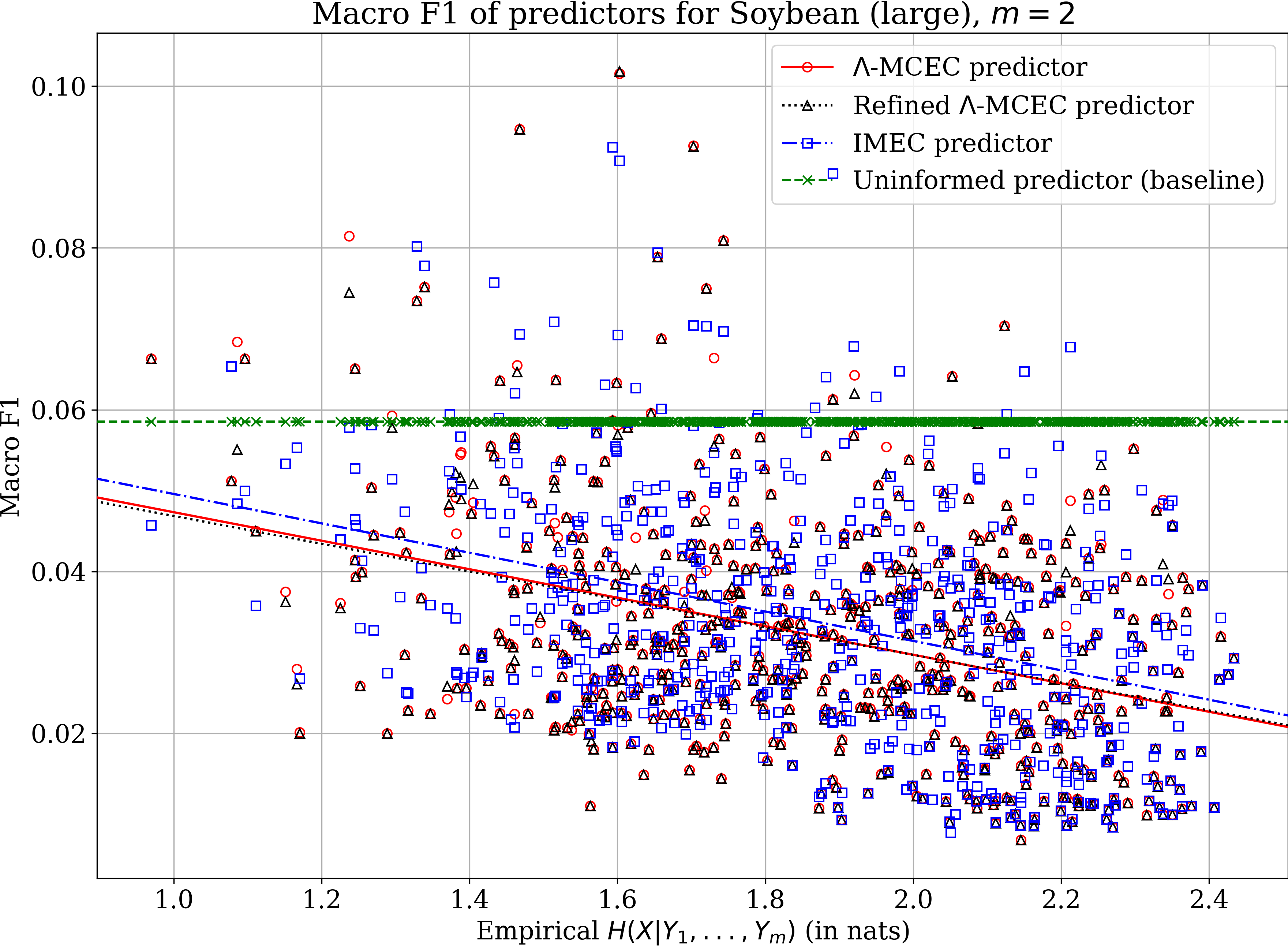}
\par\end{centering}
% \medskip{}

\caption{Comparison between $\varLambda$-MCEC, IMEC and uninformed predictors for various datasets.}\label{fig:compare_datasets}
\end{figure*}

We evaluate the MCEC and IMEC predictors using the datasets Bank Marketing \cite{moro2014data}, Car Evaluation \cite{bohanec1988knowledge}, Mushroom \cite{uci_mushroom}, Nursery \cite{olave1989application} and Soybean (large) \cite{michalski1988soybean} downloaded from the UCI Machine Learning Repository \cite{uci_repository}.\footnote{\label{footnote:license}The datasets are licensed under a Creative Commons Attribution 4.0 International (CC BY 4.0) license.
We only use datasets consisting of mostly categorical variables, since this paper focuses on categorical $X,Y_{1},\ldots,Y_{m}$.} We fix the number of feature variables to $m=2$. If the actual number of feature variables in the dataset is $m'$, we evaluate the predictors over all $\binom{m'}{m}$ combinations of subsets of features (features outside of the $m$ selected ones are discarded). Figure \ref{fig:compare_datasets} plots the macro F1 metric of the predictors for the datasets against the empirical $H(X|\mathbf{Y})$, where each datapoint is one of the $\binom{m'}{m}$ combinations of feature variables. Note that different combinations of features give different $H(X|\mathbf{Y})$. The plots also include the least-squares regression lines for the points for each predictor.
Details of the experiments are in Appendix \ref{app:datasets}.

% To summarizes the findings in Figure \ref{fig:compare_datasets}:
% \begin{itemize}
% \item 

The performance of $\varLambda$-MCEC and IMEC is generally negatively correlated with $H(X|\mathbf{Y})$ (except the Car Evaluation dataset), due to the reliance on the predictability assumption.
% \item 
$\varLambda$-MCEC and IMEC can outperform the uninformed predictor when $H(X|\mathbf{Y})$ is small ($\lessapprox 0.3$). They perform poorly for Car Evaluation and Soybean (large) datasets since $H(X|\mathbf{Y})$ is large.
% \item 
$\varLambda$-MCEC and IMEC are comparable, and they usually outperform refined $\varLambda$-MCEC.
% \end{itemize}
We also evaluate the $\varLambda$-MCEC lower bound \eqref{eq:Lambda_MCEC_lb}. See Appendix \ref{app:MCEC_bounds}.

\section{Conclusion and Limitation\label{sec:limitations}}

We proposed the $\varLambda$-MCEC predictor, which can be trained on totally unpaired datasets via linear programming and greedy minimum entropy coupling \cite{kocaoglu2017greedy}. Experiments showed that it can perform well when the conditional entropy is small.
We mention some limitations and future extensions.

% \begin{itemize}
% \item 

\paragraph{Reliance on the predictability assumption.} As the experiments on real datasets show, the $\varLambda$-MCEC predictor may perform worse than the uninformed predictor when $H(X|\mathbf{Y})$ is large. The $\varLambda$-MCEC predictor is applicable if we know that $X$ is almost a function of $\mathbf{Y}$, but may not be a good choice otherwise.
As the uninformed predictor is optimal when no relation between $X$ and $\mathbf{Y}$ is known, we cannot expect a predictor to uniformly outperform the uninformed predictor.

\paragraph{Ambiguity and class imbalance.}
As mentioned in the introduction, the $\varLambda$-MCEC predictor requires class imbalance since it cannot distinguish classes with the same probability. Even with class imbalance, it is possible that there are different couplings that achieve $H(X|\mathbf{Y})=0$, causing an ambiguity. This is a fundamental limitation due to having only totally unpaired data, that cannot be improved unless we have paired data or make additional assumptions.

% \item 
\paragraph{Robustness.} 
The $\varLambda$-MCEC predictor is an optimistic ``best-case'' predictor which considers the coupling with the smallest $H(X|\mathbf{Y})$. Robust predictors often consider the worst-case distribution instead. For the problem of learning from totally unpaired data, the worst-case distribution is simply the distribution where $X$ is independent of $Y_1,\ldots,Y_m$, where no meaningful prediction can be made. Nevertheless, we may adopt a hybrid approach where we find a predictor which optimizes the performance on the worst-case distribution among distributions with conditional entropy close to the best-case distribution, which may improve the robustness of the predictor. This is left for future studies. 

% \item 
\paragraph{Computational complexity.}
The $\varLambda$-MCEC algorithm has a running time polynomial in the dimension of the output coupling, which is exponential in the input size. An exponential running time is needed if we output the coupling as a probability table. It might be possible to reduce the running time by having a compact representation of the coupling. 
% This is left for future studies.

% \end{itemize}

\paragraph{Gap from optimality.} The $\varLambda$-MCEC algorithm computes \eqref{eq:mcec} up to a $O(\log \log N)$ gap. We may study whether the gap can be improved to $O(1)$ (which is possible for minimum entropy coupling \cite{cicalese2019minimum,rossi2019greedy,li2021efficient,compton2022tighter,compton2023minimum}, i.e., the special case $m=1$) or even arbitrarily small (possible for minimum entropy coupling
\cite{compton2026efficient}).

\section{Acknowledgement}

The work of C.\ T.\ Li was partially supported by a grant from the Research Grants Council of the Hong Kong Special Administrative Region, China [Project No: CUHK 14209823 (GRF)]. The work of S.\ Asoodeh was partially supported by Natural Sciences and Engineering Research Council of Canada.

\bibliographystyle{IEEEtran}
\bibliography{ref}

\appendices

\section{Greedy Minimum Entropy Coupling and Causal Inference} \label{app:greedy}

We briefly review the \emph{greedy minimum entropy coupling algorithm} in \cite{kocaoglu2017greedy}, which is used in the $\varLambda$-MCEC and IMEC algorithms in Section \ref{subsec:algo}. Given $\mathbf{p}\in\mathbb{R}^{n}$, $\mathbf{q}\in\mathbb{R}^{\ell}$ as probability vectors, a coupling is a matrix $\mathbf{R}\in\mathbb{R}_{\ge 0}^{n\times\ell}$ with $\mathbf{R}\mathbf{1}_{\ell}=\mathbf{p}$ and $\mathbf{R}^{\top}\mathbf{1}_{n}=\mathbf{q}$, where $\mathbf{1}_{\ell}\in\mathbb{R}^{\ell}$ is the vector with all entries being $1$. The greedy algorithm iteratively finds $i:=\mathrm{argmax}_{i}\mathbf{p}_{i}$, $j:=\mathrm{argmax}_{j}\mathbf{q}_{j}$, $r:=\min\{\mathbf{p}_{i},\mathbf{q}_{j}\}$ (which is the largest possible probability mass in $\mathbf{R}$), set $\mathbf{R}_{i,j}=r$, and repeat the algorithm on $\mathbf{p}\leftarrow\mathbf{p}-r\mathbf{e}_{i}$, $\mathbf{q}\leftarrow\mathbf{q}-r\mathbf{e}_{j}$ (where $\mathbf{e}_{i}$ is the $i$-th standard basis vector) to find the remaining entries of $\mathbf{R}$. It was shown in \cite{rossi2019greedy} that the greedy algorithm gives an almost-optimal coupling that has an entropy within $\log2$ from the optimum. Refer to \cite{li2021efficient,compton2022tighter,compton2023minimum,shkel2023information,compton2026efficient} for other bounds and algorithms for the minimum entropy coupling problem.

In \emph{entropic causal inference}~\cite{kocaoglu2017entropic,kocaoglu2017greedy}, given jointly-distributed random variables $X,Y$, we find a random variable $E$ (the \emph{exogenous variable}) that minimizes $H(E)$ under the constraint that $I(E;Y)=H(X|Y,E)=0$. The optimal value of $H(E)$ indicates how likely the causal direction $Y \to X$ is correct. This can be solved by applying minimum entropy coupling on the distributions $P_{X|Y=y}$ \cite{kocaoglu2017entropic}. It was shown in \cite{compton2020entropic} that, under a \emph{low-exogenous entropy assumption} that $H(E)$ is small for the true causal direction, the true causal direction can be identified under certain conditions. While the low-exogenous entropy assumption appears similar to the predictability assumption in this paper, their definitions are different. The predictability assumption considers the conditional entropy directly without an additional exogenous variable. If $H(E)$ is small, then $H(X|Y)$ is small, but the other direction is not necessarily true, so they are not equivalent. Also, in entropic causal inference, the joint (empirical) distribution of $X,Y$ is known. In comparison, in the learning with totally unpaired data setting in this paper, we only know the marginal distributions.

\section{Missing Proofs of Properties of $H^*$} \label{app:proofs}
\subsection{Proof of Proposition~\ref{prop:properties}}
    \begin{enumerate}
        \item For any random variables $X\sim P$ and $Y^m\coloneqq (Y_1, \dots, Y_m)\sim R\in \calC(Q_1, \dots, Q_m)$, we know that $H(X|Y^m)\leq H(X)=H(P)$. Thus, $H^*(P\Vert Q_1, \dots, Q_m)\leq H(P)$.
        \item We prove the monotonicity. Let $\pi$ be any feasible coupling for $H^*(P\|Q_1,\ldots,Q_m)$, with variables $(X,Y^m)$. Append an additional variable
$Y_{m+1}\sim Q_{m+1}$ independently of $(X,Y^m)$. The resulting
distribution is feasible for $H^*(P\|Q_1,\ldots,Q_m,Q_{m+1})$. Moreover,
$$H(X|Y^{m+1})\le H(X| Y^m) = H^*(P\|Q_1, \dots, Q_m).$$
Taking the infimum over possible couplings of $(P, Q_1, \dots, Q_{m+1})$ gives
$$H^*(P\|Q_1,\ldots,Q_m,Q_{m+1})\le H^*(P\|Q_1,\ldots,Q_m).$$
The proof for strict monotonicity is more involved. Without loss of generality, we assume $P$ has full support. Let $X\sim P$ and $Y_i \sim Q_i$ for $i=1,2,\ldots,m+1$. Assume $H(X|Y^m)=H^*(P\Vert Q_1,\ldots,Q_m)$. Moreover, let $Y^m-X-Y_{m+1}$ form a Markov chain. If $H(X|Y^{m+1})<H(X|Y^m)$, then the desired conclusion is proved. Otherwise, $X-Y^m-Y_{m+1}$ must form a Markov chain.

By double Markovity $Y^m-X-Y_{m+1}$ and $X-Y^m-Y_{m+1}$, there exist functions $f$ and $g$ such that $f(X)=g(Y^m)$ with probability $1$ and $Y_{m+1}$ is conditionally independent of $(X,Y^m)$ given $f(X)$; see \cite[page 402, Problem 25]{csiszar2011information}. Note that if $f$ is injective, then $f(X)=g(Y^m)$ with probability 1 implies $H^*(P\Vert Q_1,\ldots ,Q_m)=0$.
Therefore, if $X$ and $Y_{m+1}$ can be coupled in such a way that $Y_{m+1}$ is not conditionally independent of $X$ given any non-injective function $f(X)$ of $X$, then we can exclude double Markovity, which in turn establishes strict monotonicity. 

To this end, we construct a coupling between $X$ and $Y_{m+1}$ to ensure  $P_{Y_{m+1}|X=x}$ is distinct for every $x$.
Without loss of generality, we assume $\mathbb E[X]=0$ (which can always be satisfied by labeling the values of $X$ by distinct real numbers, and subtracting the mean from $X$). Let $A$ be a subset of the support of $Q_{m+1}$ such that $a\coloneqq Q_{m+1}(A)\in(0,1)$; this subset $A$ can always be found since $Q_{m+1}$ is not deterministic.

Define
$$P_{X,Y_{m+1}}(x, y_{m+1})=P(x)Q_{m+1}(y_{m+1})(1+\varepsilon x/a),$$ for $y_{m+1}\in A$ and  
$$P_{X,Y_{m+1}}(x, y_{m+1})=P(x)Q_{m+1}(y_{m+1})(1-\varepsilon x/(1-a)),$$ for $y_{m+1}\in A^c$, for $\varepsilon$ small enough for the distributions to be valid.
It can be verified that the constructed $P_{X,Y_{m+1}}$ has marginals $P$ and $Q_{m+1}$. Moreover, $P_{Y_{m+1}|X=x}$ is distinct for every $x$ since $P_{Y_{m+1}|X}(y_{m+1}|x)=Q_{m+1}(y_{m+1})(1+\varepsilon x/a)$ for $y_{m+1}\in A$ and $P(y_{m+1}|x)=Q_{m+1}(y_{m+1})(1-\varepsilon x/(1-a))$ for y$_{m+1}\in A^c$. 
This completes the proof. 
\item Every  feasible coupling $\pi$ of $(X,Y^m)$ induces a distribution
$R\in\mathcal C (Q_1,\ldots,Q_m)$ on $Y^m$. Viewing $S=Y^m$, the coupling $\pi$ is also a coupling of $X\sim P$ and $S\sim R$, and thus
$$H(X| Y^m)=H(X| S).$$
Thus
$$ H^*(P\|Q_1,\ldots,Q_m) \ge \min_{R\in\mathcal C (Q_1,\ldots,Q_m)}
\min_{\pi\in\mathcal C(P,R)} H_\pi(X|S), $$
where $\mathcal C(P,R)$ is the collection of all couplings between $P$ and $R$. 
Conversely, take any $R\in\mathcal C (Q_1,\ldots,Q_m)$ and any
$\pi\in\mathcal C(P,R)$. Since $R$ is a distribution on
$\calY_1\times\cdots\times\calY_m$, the random variable
$S\sim R$ can be written as $S=Y^m$, and its coordinates have
marginals $Q_i$. Hence the coupling $\pi$ defines a feasible joint distribution of
$(X,Y^m)$, with
$$H_\pi(X| S)=H_\pi(X| Y^m), $$
implying 
$$ H^*(P\|Q_1,\ldots,Q_m) \le \min_{R\in\mathcal C (Q_1,\ldots,Q_m)}
\min_{\pi\in\mathcal C(P,R)} H_\pi(X| S). $$
        \item  First suppose $H^*(P\|Q_1,\ldots,Q_m)=0$. Since the alphabets are finite, the
minimum is attained, say by $\pi^*$. Let $(X,Y^m)\sim \pi^*$, and set $R$ is the corresponding distribution of $Y^m$. Then
$R\in\mathcal C (Q_1,\ldots,Q_m)$ and $H(X| Y^m)=0.$
On finite alphabets, this is equivalent to the existence of a deterministic
map $f$ such that $X=f(Y^m)$ almost surely. Therefore
$f_\#R=P$.

Conversely, suppose there exist $R\in\mathcal C (Q_1,\ldots,Q_m)$ and a
deterministic map $f$ such that $f_\#R=P$. Draw
$S=Y^m\sim R$ and set $X=f(S)$. Then $X\sim P$, each
$Y_i\sim Q_i$, and $X$ is a deterministic function of $Y^m$. Hence
$H(X| Y^m)=0.$ Thus $H^*(P\|Q_1,\ldots,Q_m)=0$.

\item Let $\pi$ be any feasible coupling for
$H^*(P\|Q_1,\ldots,Q_m)$ and $(X,Y^m)\sim \pi$. Set $Z=g(X)$. Then $Z\sim g_\#P=P'$, while the
marginals of $Y_i$ remain $Q_i$. Hence $H^*(P'\|Q_1,\ldots,Q_m)\leq H(Z|Y^m)$. Moreover, since $Z$ is a deterministic function of $X$,
$$H(Z| Y^m) \le H(X| Y^m).$$
Taking infimum from both sides over feasible couplings $\pi$, we obtain 
$$H^*(P'\|Q_1,\ldots,Q_m) \leq  H^*(P\|Q_1,\ldots,Q_m).$$

\item Let $\pi$ and $\pi'$ be optimal couplings for $H^*(P\|Q_1,\ldots,Q_m)$ and $H^*(P'\|Q'_1,\ldots,Q'_m)$, respectively.  Consider the product coupling $\Pi=\pi\otimes\pi'$. Then, the first marginal of $\Pi$ is $P\otimes P'$ and its $(i+1)$-th marginal is $Q_i\otimes Q_i'$. Hence it is feasible for $H^*(P\otimes P'\|Q_1\otimes Q'_1,\ldots,Q_m\otimes Q'_m)$. Since $(X,Y^m)$ and $(X',(Y')^m)$ are independent under  $\Pi$, we can write 
$$ H_{\Pi}(X,X'| Y^m,(Y')^m) = H_\pi(X| Y^m)+H_{\pi'}(X'| (Y')^m). $$
Taking the infimum over feasible couplings on the left gives the desired
inequality.
\item Fix any feasible coupling and let $(X,Y^m)$ be a sample from it. We first prove the entropy-budget lower bound. Since $H(X| Y^m)=H(P)-I(X;Y^m)$ and $I(X;Y^m)\le H(Y^m)$, we obtain
$H(X| Y^m)\ge H(P)-H(Y^m)$. By subadditivity,
$H(Y^m)\le \sum_{i=1}^m H(Y_i)=\sum_{i=1}^m H(Q_i)$, hence
$$H(X| Y^m)\ge H(P)-\sum_{i=1}^m H(Q_i).$$
Since conditional entropy is always nonnegative, this yields
$$H(X| Y^m)\ge \Bigl[H(P)-\sum_{i=1}^m H(Q_i)\Bigr]_+.$$
Taking the infimum over all admissible couplings gives
$$H^*(P\|Q_1,\dots,Q_m)\ge \Bigl[H(P)-\sum_{i=1}^m H(Q_i)\Bigr]_+.$$

We now prove the second lower bound. Since each $Y_i$ takes values in $\supp(Q_i)$, the joint variable $Y^m$ takes at most $K=\prod_{i=1}^m |\supp(Q_i)|$ distinct values. Let
$$\kappa_{\mathsf{MAP}}(X| Y^m)\coloneqq \sum_{y^m} \max_x ~\Pr(X=x,Y^m=y^m)$$
denote the optimal MAP success probability for estimating $X$ from $Y^m$. For each $y^m$, choose one maximizer $\hat x(y^m)\in\arg\max_x \Pr(X=x,Y^m=y^m)$. Then
$\kappa_{\mathsf{MAP}}(X| Y^m)=\sum_{y^m}\Pr(X=\hat x(y^m),Y^m=y^m)$, and therefore
$$\kappa_{\mathsf{MAP}}(X| Y^m)\le \sum_{y^m}\Pr(X=\hat x(y^m)).$$
The set $\{\hat x(y^m):y^m\in\supp(Y^m)\}$ contains at most $K$ elements, so the last sum is at most the sum of the $K$ largest atoms of $P$, namely
\begin{equation}\label{MAPUB1}
    \kappa_{\mathsf{MAP}}(X| Y^m)\le \sum_{j=1}^{K} P^{\downarrow}(j).
\end{equation}
Since, for each realization $y^m$, the Shannon entropy dominates the min-entropy, we can write 
\begin{align*}
    H(X| Y^m)&\ge \mathbb E\left[-\log \max_x \Pr(X=x| Y^m)\right]\\
    & \ge -\log\Bigl(\sum_{y^m}\max_x \Pr(X=x,Y^m=y^m)\Bigr)\\
    & = -\log \kappa_{\mathsf{MAP}}(X| Y^m),
\end{align*}
where the second step is due to the Jensen's inequality. Combining this with \eqref{MAPUB1}  we obtain
$$H(X| Y^m)\ge -\log\Bigl(\sum_{j=1}^{K} P^{\downarrow}(j)\Bigr).$$
Taking the infimum over all admissible couplings, we conclude that
$$H^*(P\|Q_1,\dots,Q_m)\ge -\log\Bigl(\sum_{j=1}^{K} P^{\downarrow}(j)\Bigr).$$
\end{enumerate}

\subsection{Proof of Proposition~\ref{prop:uniform}}

Since $r\ge |\mathrm{supp}(P)|$, there exists an injective map
$\varphi:\mathrm{supp}(P)\to [r]$. 
Draw $X\sim P$, and set $U=\varphi(X)$. Conditional on $X$, draw $Y_1$
uniformly from $[r]$, and define
$$Y_2\coloneqq U-Y_1 \pmod r.$$
Then $Y_1\sim\Unif([r])$. Also, conditional on any value of $X$, the random
variable $Y_2$ is a translate of $-Y_1$ modulo $r$, and is therefore uniform on
$[r]$. Hence $Y_2\sim\Unif([r])$ marginally.

Moreover,
$$Y_1+Y_2=U \pmod r.$$
Since $\varphi$ is injective, $X$ is a deterministic function of $U$, and hence
a deterministic function of $(Y_1,Y_2)$. Therefore
$$H(X|Y_1,Y_2)=0.$$
This gives a feasible coupling with zero conditional entropy. Since conditional
entropy is nonnegative, it follows that
$$H^*(P\Vert \Unif([r]),\Unif([r]))=0.$$
 
\subsection{Proof of Theorem~\ref{them:binary_hetero}}
 
The proof relies on the following fact which will be proved later: \\

\noindent\textbf{Fact.} If
$q_1,\ldots,q_m\in[0,1/2]$ satisfy $t\ge 2\max_i q_i,$
then for every $a\in[0,\min\{t,1/2\}]$ there exists a coupling
$(Y_1,\ldots,Y_m)$ with $Y_i\sim\Ber(q_i)$ such that
\begin{equation}\label{fact:parity_filling}
    \Pr\big(\sum_{i=1}^m Y_i \text{ is odd}\big)=a.
\end{equation}

Relabel the coordinates so that $q_1=q^*$. Write
$$r\coloneqq \sum_{i=2}^m q_i=t-q^*,\qquad d\coloneqq q^*-r=2q^*-t.$$

\textbf{Case I:} Suppose $d\le0$ or equivalently $t\ge2q^*$. We invoke the above fact to prove
$$H^*(P\Vert Q_1,\ldots,Q_m)=\begin{cases}0, & p\le t,\\ (1-t)h_{\mathsf b}\!\big(\frac{p-t}{1-t}\big), & p>t.\end{cases}$$
Suppose first that $p\le t$. By the above fact, there exists a
coupling of $Y_1,\ldots,Y_m$ with $Y_i\sim\Ber(q_i)$ such that
$$\Pr\Big(\sum_{i=1}^mY_i \text{ is odd}\Big)=p.$$
Define $X\coloneqq \sum_{i=1}^mY_i \!\!\pmod 2.$
Then $X\sim\Ber(p)$, and $X$ is a deterministic function of $Y^m$. Hence $H(X|Y^m)=0$ and thus 
$$H^*(\Ber(p)\Vert\Ber(q_1),\ldots,\Ber(q_m))=0.$$

Now suppose that $p>t$. Let $(X,Y_1,\ldots,Y_m)$ be any coupling with
$X\sim\Ber(p)$ and $Y_i\sim\Ber(q_i)$. Set
$$A\coloneqq \{Y^m\ne0^m\}.$$
By the union bound, $\Pr(A)\le \sum_{i=1}^m\Pr(Y_i=1)=t.$
Let
$$\alpha\coloneqq \Pr(Y^m=0^m), \qquad \beta\coloneqq \Pr(X=1,Y^m=0^m).$$
Then $\alpha\ge1-t$. Moreover,
$$\beta=\Pr(X=1)-\Pr(X=1,A) \ge p-\Pr(A)\ge p-t.$$
Therefore
$$H(X|Y^m)\ge \alpha\,h_{\mathsf b}\!\Big(\frac{\beta}{\alpha}\Big).$$
Write $T\coloneqq1-t$ and $\delta\coloneqq p-t.$
Then $\alpha\ge T$ and $\beta\ge\delta$. It can be verified that, for fixed $\beta$, the function
$$u\mapsto u\,h_{\mathsf b}\!\Big(\frac{\beta}{u}\Big)$$
is nondecreasing for $u\ge\beta$. Consequently, 
$$\alpha\,h_{\mathsf b}\!\Big(\frac{\beta}{\alpha}\Big)
\ge T\,h_{\mathsf b}\!\Big(\frac{\beta}{T}\Big).$$
Since $p>t$ and $p\le1/2$, we have $t<1/2$, so $T>1/2$. Also
$\delta\le \beta\le p$ and $p\le 1-p=T-\delta.$ Thus, we can write 
$\frac{\beta}{T}\in[\frac{\delta}{T},1-\frac{\delta}{T}].$
By symmetry of $h_{\mathsf b}$ and its monotonicity on $[0,1/2]$,
$$h_{\mathsf b}\!\Big(\frac{\beta}{T}\Big)\ge h_{\mathsf b}\!\Big(\frac{\delta}{T}\Big).$$
Therefore every feasible coupling satisfies
$$H(X|Y^m) \ge (1-t)h_{\mathsf b}\!\Big(\frac{p-t}{1-t}\Big).$$
This lower bound is attained by the following coupling. Put
$$\Pr(Y^m=e_i)=q_i,\qquad i\in[m],$$
and
$$\Pr(Y^m=0^m)=1-t,$$
with no mass on the remaining points of $\{0,1\}^m$. Then $Y_i\sim\Ber(q_i)$ for every $i$. Set $X=1$ on each singleton atom $e_i$, and on the atom $0^m$
let
$$\Pr(X=1|Y^m=0^m)=\frac{p-t}{1-t}.$$
Then $X\sim\Ber(p)$, and the only conditional uncertainty in $X$ occurs at $Y^m=0^m$. Hence
$$H(X|Y^m)=(1-t)h_{\mathsf b}\!\Big(\frac{p-t}{1-t}\Big).$$

\textbf{Case II:} We now consider the case $d>0$, or equivalently $q^*>r$. We will prove that 
$$H^*(P\Vert Q_1,\ldots,Q_m)=
\begin{cases}
0, & p\in[0,r]\cup[d,t],\\
d\min\left\{h_{\mathsf b}\!\left(\frac{p}{d}\right),h_{\mathsf b}\!\left(\frac{p-r}{d}\right)\right\}, & r<p<d,\\
(1-t)h_{\mathsf b}\!\left(\frac{p-t}{1-t}\right), & p>t.
\end{cases}$$

The case $p>t$ is precisely similar to what we proved in Case I, so we ignore it here.  We now turn to the zero-entropy cases. Suppose $p\le r$. Choose numbers $u_i\in[0,q_i]$, $i=2,\ldots,m$, such that $\sum_{i=2}^m u_i=p$. Define a distribution of $Y^m$ by putting mass $u_i$ on $e_i$, mass $q_i-u_i$ on $e_1+e_i$, mass $d+p$ on $e_1$, and the remaining mass on $0^m$. The $i$th marginal is $q_i$ for every $i\ge2$, while the first marginal is
$$d+p+\sum_{i=2}^m(q_i-u_i)=d+p+r-p=q^*.$$
Now set $X=1$ exactly on the atoms $e_i$, $i=2,\ldots,m$. This event has probability $p$, so $X\sim\Ber(p)$ and $X$ is a deterministic function of $Y^m$. Hence $H(X|Y^m)=0$.

Next suppose $d\le p\le q^*$. Choose numbers $v_i\in[0,q_i]$, $i=2,\ldots,m$, such that $\sum_{i=2}^m v_i=q^*-p$. Put mass $v_i$ on $e_1+e_i$, mass $q_i-v_i$ on $e_i$, mass $p$ on $e_1$, and the remaining mass on $0^m$. The first marginal is
$$p+\sum_{i=2}^m v_i=p+q^*-p=q^*,$$
and the other marginals are again $q_i$. Setting $X=1$ exactly on the atom $e_1$ gives $X\sim\Ber(p)$ and $H(X|Y^m)=0$.

Finally suppose $q^*\le p\le t$. Choose numbers $u_i\in[0,q_i]$, $i=2,\ldots,m$, such that $\sum_{i=2}^m u_i=p-q^*$. Put mass $u_i$ on $e_i$, mass $q_i-u_i$ on $e_1+e_i$, mass $p-r$ on $e_1$, and the remaining mass on $0^m$. The first marginal is
$$p-r+\sum_{i=2}^m(q_i-u_i)=p-r+r-(p-q^*)=q^*,$$
and the other marginals are $q_i$. Now set $X=1$ on all atoms with first coordinate equal to $1$, and also on the atoms $e_i$, $i=2,\ldots,m$. The probability of this event is
$$q^*+(p-q^*)=p.$$
Thus again $X\sim\Ber(p)$ and $H(X|Y^m)=0$. Combining the three zero-entropy constructions gives
$$H^*(P\Vert Q_1,\ldots,Q_m)=0\qquad\text{for }p\in[0,r]\cup[d,t].$$

We now consider the gap $r<p<d$. Let $(X,Y_1,\ldots,Y_m)$ be any feasible coupling. Set
$$B\coloneqq\{Y^m=e_1\},\qquad C\coloneqq\{Y^m\notin\{0^m,e_1\}\}.$$
Since every point in $C$ has at least one of the coordinates $Y_2,\ldots,Y_m$ equal to $1$, the union bound gives
$$\Pr(C)\le \sum_{i=2}^m q_i=r.$$
Moreover,
$$\Pr(B)=\Pr(Y_1=1,Y_2=\cdots=Y_m=0)\ge q^*-\Pr(C)\ge q^*-r=d.$$
Also,
$$\Pr(Y^m=0^m)\ge 1-t\ge d,$$
where the last inequality follows from $q^*\le1/2$.

Let $\theta=\Pr(X=1,C)$. Since $\Pr(C)\le r$, we have $0\le\theta\le r$. The remaining mass $p-\theta$ of the event $\{X=1\}$ must lie over the two atoms $0^m$ and $e_1$, each of which has mass at least $d$. For a fixed amount $L\in[0,d]$, the smallest entropy cost of placing $L$ units of $X=1$ on atoms of mass at least $d$ is at least $d\,h_{\mathsf b}(L/d)$. Therefore
$$H(X|Y^m)\ge d\,h_{\mathsf b}\!\Big(\frac{p-\theta}{d}\Big).$$
Since $\theta\in[0,r]$, we have $p-\theta\in[p-r,p]$. The binary entropy is concave, so its minimum over this interval is attained at one of the endpoints. Hence every feasible coupling satisfies
$$H(X|Y^m)\ge d\min\Big\{h_{\mathsf b}\!\Big(\frac{p}{d}\Big),h_{\mathsf b}\!\Big(\frac{p-r}{d}\Big)\Big\}.$$

This lower bound is attained as follows. Put mass $d$ on $e_1$, mass $q_i$ on $e_1+e_i$ for each $i=2,\ldots,m$, and the remaining mass on $0^m$. This gives the correct marginals, because the first marginal is $d+r=q^*$ and the $i$th marginal is $q_i$ for $i\ge2$.

There are two choices. First, put no $X=1$ mass on the atoms $e_1+e_i$, and randomize on $e_1$ by setting
$$\Pr(X=1|Y^m=e_1)=\frac{p}{d}.$$
This gives entropy $d\,h_{\mathsf b}(p/d)$. Second, set $X=1$ on all atoms $e_1+e_i$, whose total mass is $r$, and randomize on $e_1$ by setting
$$\Pr(X=1|Y^m=e_1)=\frac{p-r}{d}.$$
This gives entropy $d\,h_{\mathsf b}((p-r)/d)$. Taking the better of the two constructions gives
$$H^*(P\Vert Q_1,\ldots,Q_m)=d\min\Big\{h_{\mathsf b}\!\Big(\frac{p}{d}\Big),h_{\mathsf b}\!\Big(\frac{p-r}{d}\Big)\Big\}$$
for $r<p<d$.

It remains to prove the above Fact.

\begin{proof}[Proof of the Fact]
The case $t=0$ is trivial, thus we  assume $t>0$. We construct the desired distribution of $Y^m=(Y_1,\ldots,Y_m)$ by mixing two components: one supported on odd-parity vectors and one supported on even-parity vectors. Here, parity means $\sum_iY_i \pmod 2$, and thus $Y^m$ has odd parity if an odd number of the coordinates are equal to $1$, and even parity otherwise.

First define a probability distribution $\mu_1$ on $\{0,1\}^m$ by putting mass $q_i/t$ on the vector $e_i$, for $i\in[m]$. Then $\mu_1$ is supported on odd-parity vectors and has one-dimensional marginals
$$\mu_1(Y_i=1)=\frac{q_i}{t}.$$
We want this odd-parity component to have total weight $a$. Suppose the remaining component has weight $1-a$ and has $i$th marginal $b_i$. In the final mixture, the $i$th marginal must be $q_i$, so we need
$$a\frac{q_i}{t}+(1-a)b_i=q_i.$$
Solving for $b_i$ gives
$$b_i\coloneqq \frac{q_i-aq_i/t}{1-a}=\frac{t-a}{t(1-a)}q_i.$$
It remains to show that the vector $b=(b_1,\ldots,b_m)$ can be realized as the marginal vector of a distribution supported on even-parity vectors.

We use the standard halfspace description of the even-parity polytope \cite{Jeroslow1975OnDS}. A vector $b\in[0,1]^m$ is the marginal vector of a distribution supported on
$$\mathcal E\coloneqq\Big\{y\in\{0,1\}^m:\sum_i y_i \text{ is even}\Big\}$$
if and only if, for every odd set $A\subseteq[m]$, meaning every set $A$ with cardinality $|A|$ odd, one has
\begin{equation}\label{eq:evenSet}
    \sum_{i\in A}b_i-\sum_{i\notin A}b_i\le |A|-1.
\end{equation}
We shall prove this inequality. To do so, first set
$$\lambda\coloneqq \frac{t-a}{t(1-a)},\qquad \text{and}\qquad b_i=\lambda q_i.$$
Since $a\le t$, we have $\lambda\ge0$. Also $b_i\le1$ for every $i$: if $t\le1$, then $\lambda\le1$; if $t>1$, then $a\le1/2$ gives $\lambda\le 2-\frac1t<2,$ and hence $b_i=\lambda q_i\le1$ because $q_i\le1/2$.
Now fix an odd set $A\subseteq[m]$. Write
$$k\coloneqq |A|,\qquad \text{and}\qquad q(A)\coloneqq\sum_{i\in A}q_i.$$

Therefore, to prove \eqref{eq:evenSet}, it is enough to prove
$$\lambda(2q(A)-t)\le k-1.$$
If $2q(A)\le t$, this is immediate. Suppose $2q(A)>t$. If $k=1$, then $A=\{i\}$ for some $i$, and $2q_i>t$, contradicting $t\ge2\max_iq_i$. Therefore $k\ge3$.

If $t\le1$, then $\lambda\le1$, and hence
$$\lambda(2q(A)-t)\le 2q(A)-t\le t\le 1\le k-1.$$
It remains to consider $t>1$. Since $a\le1/2$,
$$\lambda=\frac{t-a}{t(1-a)}\le 2-\frac1t.$$
Moreover,
$$2q(A)-t\le \min\{t,k-t\},$$
because $q(A)\le t$ and $q(A)\le k/2$. If $t\le k/2$, then
$$\lambda(2q(A)-t)\le\Big(2-\frac1t\Big)t=2t-1\le k-1.$$
If $t\ge k/2$, then
$$\lambda(2q(A)-t)\le\Big(2-\frac1t\Big)(k-t)\le k-1.$$
The last inequality follows because the function $u\mapsto(2-1/u)(k-u)$ is nonincreasing for $u\ge k/2$ and equals $k-1$ at $u=k/2$.
Thus, $b$ satisfies all the even-parity inequalities. Hence there exists a distribution $\mu_0$ supported on even-parity vectors with marginals $b_i$.

Finally define $\mu\coloneqq a\mu_1+(1-a)\mu_0.$ Under $\mu$, the $i$th coordinate has marginal
$$a\frac{q_i}{t}+(1-a)b_i=a\frac{q_i}{t}+(1-a)\frac{t-a}{t(1-a)}q_i=q_i.$$
Therefore, if $Y^m\sim\mu$, then $Y_i\sim\Ber(q_i)$ for every $i$. Moreover, $\mu_1$ is supported on odd-parity vectors and $\mu_0$ is supported on even-parity vectors, hence we have 
$$\Pr\Big(\sum_{i=1}^mY_i \text{ is odd}\Big)=a,$$
which completes the proof.
\end{proof}

\subsection{Proof of Theorem~\ref{thm:heterogeneous-exp}}
We use the quantile coupling, or inverse-CDF, construction. Generate both
$X\sim P$ and the tuple $Y^m\coloneqq (Y_1,\ldots,Y_m)\sim
Q_1\otimes\cdots\otimes Q_m$ from a common uniform random variable
$U\sim\mathrm{Unif}[0,1)$ using their interval representations. The partition for
$X$ has at most $r-1$ interior boundary points. Each realization of $Y^m$
corresponds to an interval whose length is at most
$\prod_{i=1}^m\|Q_i\|_\infty$. If such an interval does not cross a boundary
of the $P$-partition, then it lies inside a single $P$-interval, so $Y^m$
uniquely determines $X$. Hence ambiguity can occur only for tuple intervals
that cross one of the $P$-partition boundaries. The total probability of this event is at most $(r-1)\prod_{i=1}^m\|Q_i\|_\infty$. 

Let $\calX=\{x_1,\ldots,x_r\}$ and construct a coupling of $X\sim P$ and
$Y_i\sim Q_i$ for $i\in[m]$. Let
$$
\calS\coloneqq \supp(Q_1)\times\cdots\times\supp(Q_m).
$$
For each $a=(a_1,\ldots,a_m)\in\calS$, define
$$
q_m(a)\coloneqq \prod_{i=1}^m Q_i(a_i).
$$
Choose an arbitrary ordering $\calS=\{a^{(1)},\ldots,a^{(N)}\}$, and partition
$[0,1)$ into half-open intervals
$$I_j\coloneqq \Bigg[\sum_{k<j}q_m(a^{(k)}),\sum_{k\le j}q_m(a^{(k)})\Bigg),
\qquad j\in[N].$$
Then $|I_j|=q_m(a^{(j)})$. If $U\sim\mathrm{Unif}[0,1)$ and we set
$$
Y^m=(Y_1,\ldots,Y_m)=a^{(j)}\quad\text{whenever }U\in I_j,
$$
then
$$
\Pr(Y^m=a)=q_m(a)=\prod_{i=1}^m Q_i(a_i).
$$
Thus $Y^m\sim Q_1\otimes\cdots\otimes Q_m$, and in particular $Y_i\sim Q_i$
for each $i\in[m]$.

Next, partition $[0,1)$ according to $P$ by setting
$$J_\ell\coloneqq \Bigg[\sum_{k<\ell}P(x_k),\sum_{k\le \ell}P(x_k)\Bigg),
\qquad \ell\in[r].$$
Then $\{J_\ell\}_{\ell=1}^r$ is a partition of $[0,1)$ up to empty intervals,
and $|J_\ell|=P(x_\ell)$. Define $X=x_\ell$ whenever $U\in J_\ell$. Then
$X\sim P$. Therefore $(X,Y_1,\ldots,Y_m)$ is a feasible coupling of
$(P,Q_1,\ldots,Q_m)$.

Let
$$\calB\coloneqq \Bigg\{\sum_{k\le\ell}P(x_k):\ell=1,\ldots,r-1\Bigg\}$$
be the set of boundary points of the $P$-partition. This set has cardinality
at most $r-1$. Define the ambiguity event
$$
\calA\coloneqq
\left\{U\in I_j\text{ for some }j\text{ such that }I_j\cap\calB\neq\varnothing\right\},
$$
and let $A=\mathbf 1_{\calA}$. Since membership of $U$ in an interval $I_j$
is determined by $Y^m$, the random variable $A$ is a function of $Y^m$.

On the event $\{A=0\}$, the interval $I_j$ containing $U$ contains no boundary
point of the $P$-partition. Hence $I_j$ is contained in a single interval
$J_\ell$ up to endpoints. Therefore, once $Y^m=a^{(j)}$ is known, the value
of $X$ is uniquely determined. Thus
$$
H(X\mid Y^m,A=0)=0.
$$

We next bound $\Pr(A=1)$. Each boundary point in $\calB$ belongs to at most
one interval $I_j$, and there are at most $r-1$ boundary points. Therefore
$$
\Pr(A=1)\le (r-1)\max_{1\le j\le N}|I_j|
=(r-1)\max_{a\in\calS}q_m(a).
$$
Moreover,
$$
\max_{a\in\calS}q_m(a)
=
\max_{a\in\calS}\prod_{i=1}^m Q_i(a_i)
\le
\prod_{i=1}^m\|Q_i\|_\infty.
$$
Hence
$$\Pr(A=1)\le \min\Big\{1,(r-1)\prod_{i=1}^m\|Q_i\|_\infty\Big\} \eqqcolon \zeta_m.$$

Since $A$ is a function of $Y^m$, we have $H(A\mid Y^m)=0$, and therefore
$H(X\mid Y^m)=H(X\mid Y^m,A)$. Thus
$$
H(X\mid Y^m)
=
\Pr(A=1)H(X\mid Y^m,A=1)
+\Pr(A=0)H(X\mid Y^m,A=0).
$$
Using $H(X\mid Y^m,A=0)=0$ and $H(X\mid Y^m,A=1)\le\log r$, we obtain
$$
H(X\mid Y^m)\le \Pr(A=1)\log r\le \zeta_m\log r.
$$
Since $H^*(P\|Q_1,\ldots,Q_m)$ is the infimum of $H(X\mid Y^m)$ over all
feasible couplings, we conclude that
$$
H^*(P\|Q_1,\ldots,Q_m)\le \zeta_m\log r.
$$
This proves the claim.

\subsection{Proof of Theorem~\ref{thm:equalZero}}

We first prove the claim under the stronger assumption $\ell=1$, namely $Q(x)>0$ for all $x\in[r]$. We then extend it to general $\ell$. Fix $n\ge 1$, and write $\mu_n:=Q^{*n}$ for the $n$-fold convolution of $Q$ on $[r]$ and $q_{\min}$ for the $\min_{x\in [r]} Q(x)$. Consider the random walk on $[r]$ with transition kernel
$K(x,y)=Q(y-x)$ where subtraction is modulo $r$. Since $Q(y-x)\ge q_{\min}$ for all $x,y$, we have the minorization
$$K(x,\cdot)\ge rq_{\min}\,\Unif([r])(\cdot),\qquad\forall x\in[r],$$
where $\Unif([r])$ is the distribution on $[r]$. Since $\Unif([r])$ is an invariant measure for $K$, the Doeblin contraction bound gives
$$\TV(\mu_n,\Unif([r]))\le (1-rq_{\min})^n.$$
Also, since $\Unif([r])$ is shift-invariant, it follows that for every $j\in[r]$,
\begin{equation}\label{eq:couplingproof1}
  \TV(\mu_n,\tau_{-j}\mu_n)
\le 2\,\TV(\mu_n,\Unif([r]))
\le 2(1-rq_{\min})^n,  
\end{equation}
where $\tau_{j}$ is the shift operator.

Now fix $j\in[r]$, and let $(A_j,B_j)$ be a maximal coupling of $\mu_n$ and $\tau_{-j}\mu_n$.
Define
\[T_j:=B_j+j \pmod r\qquad \text{and} \qquad D_j:=T_j-A_j \pmod r.\]
Then both $A_j$ and $T_j$ have distribution $\mu_n$, and on the event $\{A_j=B_j\}$ we have $D_j=j$.
Hence
\begin{align*}
\TV(\mathsf{Law}(D_j),\delta_j)
&= \Pr(D_j \neq j)\\
&\le \Pr(A_j\neq B_j)\\
&= \TV(\mu_n,\tau_{-j}\mu_n)\\
&\le 2(1-rq_{\min})^n.
\end{align*}

Let $M_n$ be the $r\times r$ matrix whose $j$th column is $\mathsf{Law}(D_j)$.
Thus we need to show that there exists $w\in \Delta([r])$ such that $M_n w = P$.
Since each column of $M_n$ is within total variation distance at most $2(1-rq_{\min})^n$ of the corresponding simplex vertex $\delta_j$, we have
\[\|M_n-I\|_1 \le 4(1-rq_{\min})^n,\]
where $\|\cdot\|_1$ is the operator norm induced by $\ell_1$ norm, i.e., the maximum column sum norm. Assume now that
\[(1-rq_{\min})^n<\frac{p_{\min}}{8},\]
with $p_{\min}>0$. Since $p_{\min}\le 1$, this implies $\|M_n-I\|_1<1/2$, so $M_n$ is invertible and
\[\|M_n^{-1}\|_1\le \frac{1}{1-\|M_n-I\|_1}\le 2.\]
Let \(w:=M_n^{-1}P.\) Then, we can write 
\begin{align*}
\|w-P\|_1 & = \|M_n^{-1}(I-M_n)P\|_1\\
& \le \|M_n^{-1}\|_1\,\|I-M_n\|_1\,\|P\|_1\\
& \le 8(1-rq_{\min})^n\\
& < p_{\min}.
\end{align*}
Therefore, for every $j\in[r]$,
\[w(j)\ge P(j)-\|w-P\|_1 >0.\]
Also, since every column of $M_n$ sums to $1$, we have \(w\in\Delta([r])\).
Thus
\[P=\sum_{j\in[r]} w(j)\mathsf{Law}(D_j).\]

We now realize each $D_j$ from a valid coupling of $2n$ copies of $Q$.
For each $s\in[r]$, let $\alpha_s$ be the conditional distribution of $(Z_1,\dots,Z_n)$ given that
\[Z_1+\cdots+Z_n\equiv s \pmod r,\]
where $Z_1,\dots,Z_n\stackrel{\text{i.i.d.}}{\sim} Q$.
This is well defined because $\mu_n(s)>0$ for all $s$.
Given $(A_j,T_j)=(s,t)$, sample $(Y_1,\dots,Y_n)\sim \alpha_s$ and $(Y_{n+1},\dots,Y_{2n})\sim \alpha_t$ independently.
Let $\pi^{(j)}$ be the resulting distribution on $([r])^{2n}$.
Although, the conditional distribution \(\alpha_s\) does not itself have one-dimensional marginals \(Q\), the unconditional first block under \(\pi^{(j)}\) has distribution
\[\sum_{s\in[r]} \Pr(A_j=s)\,\alpha_s=\sum_{s\in[r]} \mu_n(s)\,\alpha_s=Q^{\otimes n},\]
by the law of total probability. Similarly, the unconditional second block has distribution
\[\sum_{t\in[r]} \Pr(T_j=t)\,\alpha_t=\sum_{t\in[r]} \mu_n(t)\,\alpha_t=Q^{\otimes n},\]
since \(T_j\sim \mu_n\). Therefore each coordinate under $\pi^{(j)}$ has marginal $Q$. 

Now define
\[f(y_1,\dots,y_{2n})\coloneqq \Bigl(\sum_{i=n+1}^{2n} y_i\Bigr)-\Bigl(\sum_{i=1}^n y_i\Bigr)
\pmod r.\]
Under $\pi^{(j)}$, the first block sum is $A_j$ and the second is $T_j$, so
\[f_\#\pi^{(j)}=\mathsf{Law}(D_j).\]

Finally, let
\[\pi:=\sum_{j\in[r]} w(j)\pi^{(j)}.\]
Then every one-dimensional marginal of $\pi$ is $Q$, and
\[f_\#\pi=P.\]
Therefore, if $(Y_1,\dots,Y_{2n})\sim \pi$ and we set \(X:=f(Y_1,\dots,Y_{2n}),\)
then $H(X| Y_1,\dots,Y_{2n})=0$.
Hence \(H^*_{2n}(P\|Q)=0.\)

If $m\ge 2n$, we may append extra coordinates with marginal $Q$ arbitrarily. Then $X$ remains a deterministic function of the full list, so $H^*_m(P\|Q)=0$ as well.
Thus, any $n_0$ satisfying
\[(1-rq_{\min})^{n_0}<\frac{p_{\min}}{8}\]
yields the conclusion with $m_0=2n_0$. Thus, one may choose $m_0$ to be 
\[m_0=2\left\lceil
\frac{\log(8/p_{\min})}{-\log(1-rq_{\min})}
\right\rceil.\]

It remains to show the result for general $\ell > 1$. If $Q$ does not have full support, let $\ell$ be such that $Q^{*\ell}$ has full support.
Group the variables into blocks of length $\ell$:
\[W_t := Y_{(t-1)\ell+1}+\cdots+Y_{t\ell}\pmod r.\]
Then each $W_t$ has distribution $Q^{*\ell}$.
Applying the above construction to $(W_1,\dots,W_{2n})$ yields exact recovery using $2n$ blocks, i.e., $2n\ell$ original samples.
This proves the claim with
\[m_0 = 2\ell n_0=2\ell \left\lceil \frac{\log(8/p_{\min})}{-\log(1-r\widetilde q_{\min})}\right\rceil.\]

\section{Missing Proofs of Properties of $\varLambda$-MCEC} \label{app:MCEC_proofs}

\begin{proof}[Proof of Theorem~\ref{thm:Lambda_le_Hstar}]
Consider a coupling $X,Y_{1},\ldots,Y_{m}$ of $P,Q_{1},\ldots,Q_{m}$. Let $R=P_{Y^{m}}$. Consider the conditional distribution $P_{U|X,Y^{m}}$ which minimizes $H(U)$ subject to the constraint that $H(X|Y^{m},U)=0$, i.e., $U$ is the conditional compression in \cite{li2026discrete}. By \cite{li2026discrete}, for every $\eta>0$,
\begin{align*}
H(U) & \le H(X|Y^{m})+\log\Big(1+\frac{H(X|Y^{m})}{e\eta}\Big)+\eta.
\end{align*}
Also, by the mixed chain rule of discrete layered entropy \cite{li2026discrete},
\begin{align*}
H(U) & \ge\varLambda(Y^{m},U)-\varLambda(Y^{m})\\
 & \ge\varLambda(X,Y^{m})-\varLambda(Y^{m})\\
 & \ge\varLambda(P\wedge R)-\varLambda(R) \, \ge \, \varLambda_{\wedge}^{*}(P\Vert(Q_{i})_{i}),
\end{align*}
where the penultimate inequality is due to the Schur-concavity of $\varLambda$ \cite{shannon1951prediction,li2026discrete}, and due to $P_{X,Y^{m}}\preceq P$ and $P_{X,Y^{m}}\preceq R$, so $P_{X,Y^{m}}\preceq P\wedge R$. The bound \eqref{eq:Lambda_MCEC_Hstar} follows. 
The optimal choice of $\eta$ which gives the best bound for \eqref{eq:Lambda_MCEC_Hstar} is $\eta=\frac{1}{2e}(\sqrt{(H^{*})^{2}+4eH^{*}\log e}-H^{*})$, by the same arguments in \cite{li2026discrete}. Hence, we have $\varLambda^*_{\wedge} \le g(H^*)$, $g^{-1}(\varLambda^*_{\wedge}) \le H^*$.
\end{proof}

\begin{proof}[Proof of Corollary~\ref{cor:Lambda_approx_Hstar}]
We prove the following slightly stronger bounds.
Letting $n_{X}:=|\mathrm{supp}(P)|$, $N_{Y}:=\prod_{i=1}^{m}|\mathrm{supp}(Q_{i})|$, $N=n_{X}N_{Y}$, we have
\[
\varLambda_{\wedge}^{*}\ge H^{*}-\log\big(1+e^{-1}\max\{\log n_{X},\log N_{Y}\}\big)-\log(2e),
\]
\[
\varLambda_{\wedge}^{*}\le H^{*}+\log(1+e^{-1}\log n_{X})+\log e.
\]
The upper bound follows from taking $\eta=1$ in Theorem \ref{thm:Lambda_le_Hstar}. For the lower bound, for $R\in\mathcal{C}(Q_{1},\ldots,Q_{m})$, by \cite{li2026discrete} and that $P\wedge R$ has at most $\max\{ n_{X}, N_{Y}\}$ nonzero entries,
\begin{align*}
 & H(P\wedge R)-H(R)\\
 & \le\varLambda(P\wedge R)+\log(1+e^{-1}\varLambda(P\wedge R))+ \log e -\varLambda(R)\\
 & \le\varLambda(P\wedge R)-\varLambda(R) \\
 & \quad +\log(1+e^{-1}\max\{\log n_{X},\log N_{Y}\})+ \log e.
\end{align*}
The result follows from $H_{\wedge}^{*}\ge H^{*}-\log2$ \cite{cicalese2019minimum}.
\end{proof}

\begin{proof}[Proof of Corollary~\ref{cor:MCEC_almost_optimal}]
In the proof of Theorem \ref{thm:Lambda_le_Hstar}, we have shown that for every coupling $\tilde{X},\tilde{Y}_{1},\ldots,\tilde{Y}_{m}$ of $P,Q_{1},\ldots,Q_{m}$, we have
\begin{equation}
H(\tilde{X}|\tilde{Y}^{m})+\log(1+e^{-1}H(\tilde{X}|\tilde{Y}^{m}))+ \log e \ge\varLambda_{\wedge}^{*}.\label{eq:pf_every_coupling}
\end{equation}
We now bound $H(X|Y^{m})$ for the coupling given by the $\varLambda$-MCEC algorithm. Recall that the coupling $R\in\mathcal{C}(Q_{1},\ldots,Q_{m})$ given by the algorithm attains the minimal $\varLambda(P\wedge R)-\varLambda(R)=\varLambda_{\wedge}^{*}$. The proof of Corollary \ref{cor:Lambda_approx_Hstar} shows that
\begin{align*}
 & H(P\wedge R)-H(R)\\
 & \le\varLambda_{\wedge}^{*}+\log(1+e^{-1}\max\{\log n_{X},\log N_{Y}\})+ \log e,
\end{align*}
where $n_{X}:=|\mathrm{supp}(P)|$, $N_{Y}:=\prod_{i=1}^{m}|\mathrm{supp}(Q_{i})|$. The greedy minimum entropy coupling algorithm \cite{kocaoglu2017greedy} finds a coupling $P_{X,Y^{m}}\in\mathcal{C}(P,R)$ with $H(P_{X,Y^{m}})\le H(P\wedge R)+\log2$ \cite{rossi2019greedy}, and hence 
\begin{align*}
H(X|Y^{m}) & =H(P_{X,Y^{m}})-H(R)\\
 & \le\varLambda_{\wedge}^{*}+\log(1+e^{-1}\max\{\log n_{X},\log N_{Y}\})+\log(2e).
\end{align*}
Combining with (\ref{eq:pf_every_coupling}), for every coupling $\tilde{X},\tilde{Y}_{1},\ldots,\tilde{Y}_{m}$ of $P,Q_{1},\ldots,Q_{m}$,
\begin{align*}
 & H(X|Y^{m})-H(\tilde{X}|\tilde{Y}^{m})\\
 & \le\log(1+e^{-1}\max\{\log n_{X},\log N_{Y}\})+\log(1+e^{-1}H(\tilde{X}|\tilde{Y}^{m}))+\log(2e^2)\\
 & \le\log(1+e^{-1}\max\{\log n_{X},\log N_{Y}\})+\log(1+e^{-1}\log n_{X})+\log(2e^2).
\end{align*}
\end{proof}

\section{Comparing More Predictors in Simulation Experiments\label{app:more}}

Continuing the simulation experiments in Section \ref{subsec:sim}, we include three additional predictors in the comparison shown in Figure \ref{fig:compare_more}:\footnote{Each data point is computed via an average over $20000$ trials. Each trial consists of $1000$ samples for the training data, and $1000$ samples for the testing data.}

\begin{itemize}
\item The plug-in \emph{maximum a posteriori (MAP) predictor} $\mathrm{argmax}_{x}\hat{P}_{X|\mathbf{Y}}(x|\mathbf{y})$, where $\hat{P}_{X|\mathbf{Y}}$ is the empirical conditional distribution. The MAP predictor requires paired training data on $(X,Y_1,\ldots,Y_m)$, whereas $\varLambda$-MCEC and IMEC require only totally unpaired data, so this is not a fair comparison. The MAP predictor is included only to show the macro F1 that can be achieved if we have paired data. In Figure \ref{fig:compare_more}, it is plotted with dotted line to indicate that it requires more data than the setting (learning with totally unpaired data) allows.
\item The \emph{naive Bayes predictor} \cite{domingos1997optimality}, which also requires paired training data, or at least the empirical joint distribution of each pair $(X,Y_i)$ for $i=1,\ldots,m$. This is also not a fair comparison. It is plotted with dashed line to indicate that it requires more data than the setting allows.
\item The \emph{refined IMEC predictor}, which is the IMEC predictor with the refinement step in Section \ref{subsec:algo}. Note that the refinement step may make $Y_1,\ldots,Y_m$ dependent, so the independent features assumption may no longer hold. This is included only for the sake of completeness.
\end{itemize}

\begin{figure*}
\begin{centering}
\includegraphics[scale=0.268]{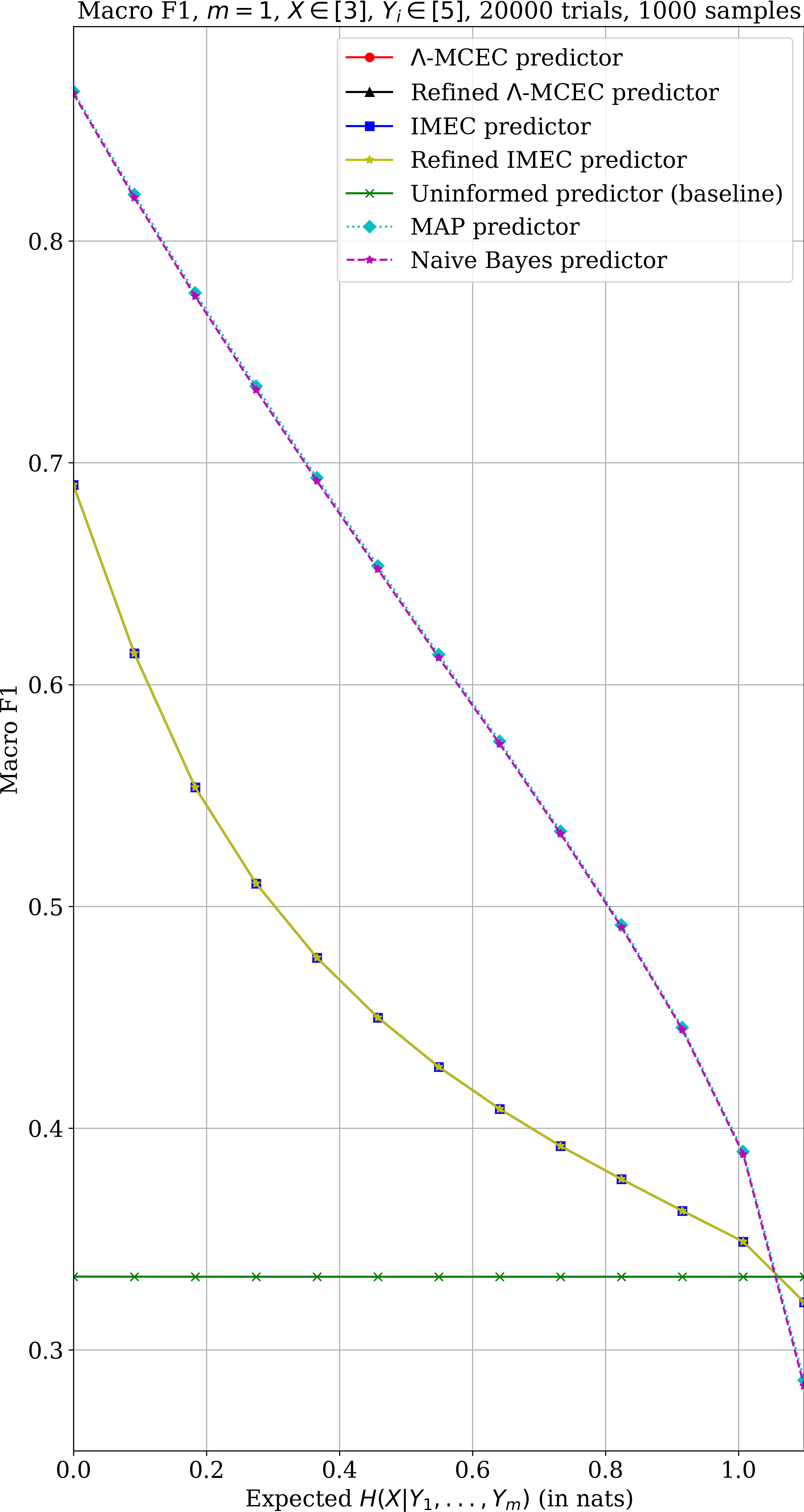} %$\;$
\includegraphics[scale=0.268]{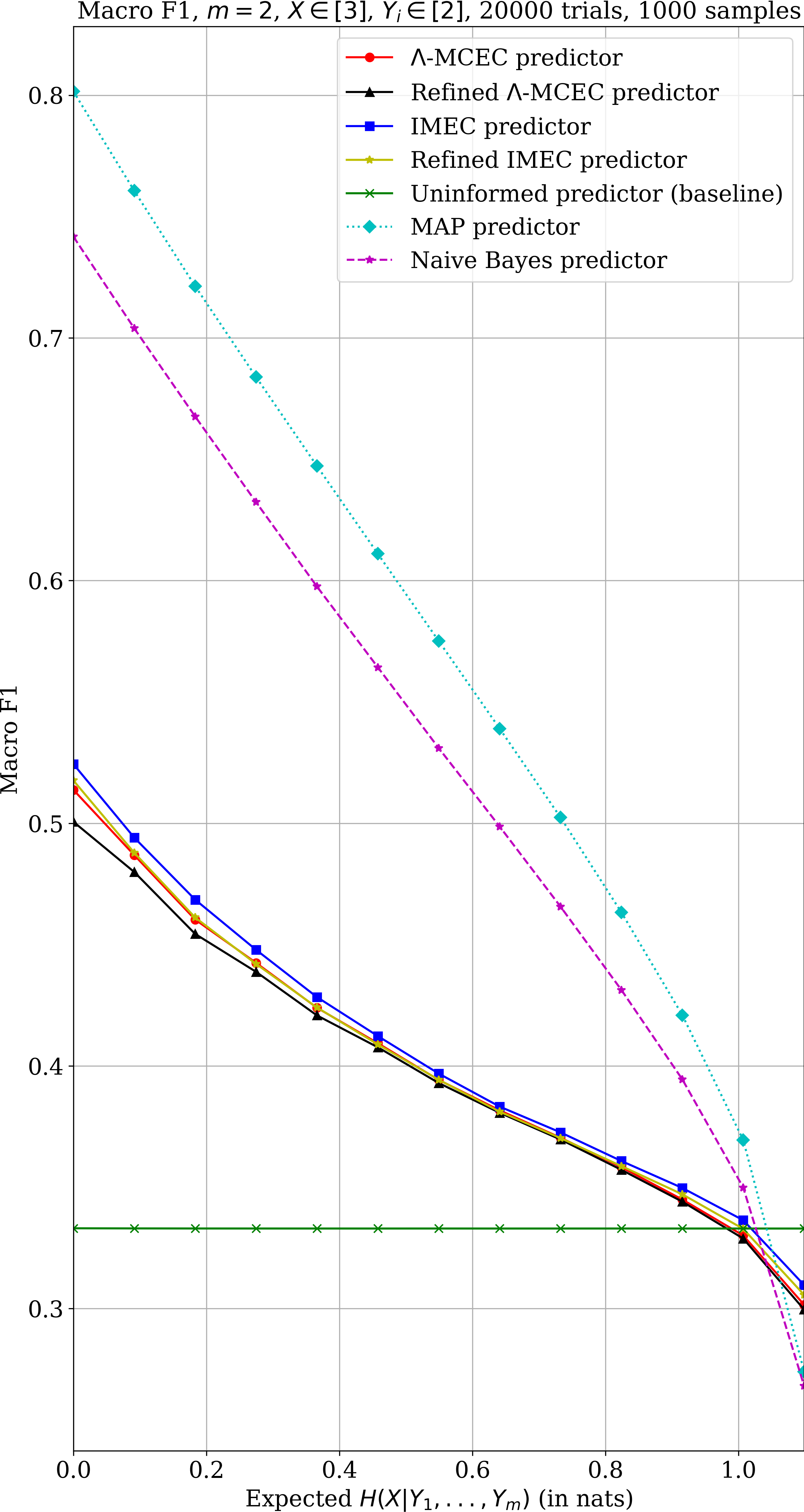} %$\;$
\includegraphics[scale=0.268]{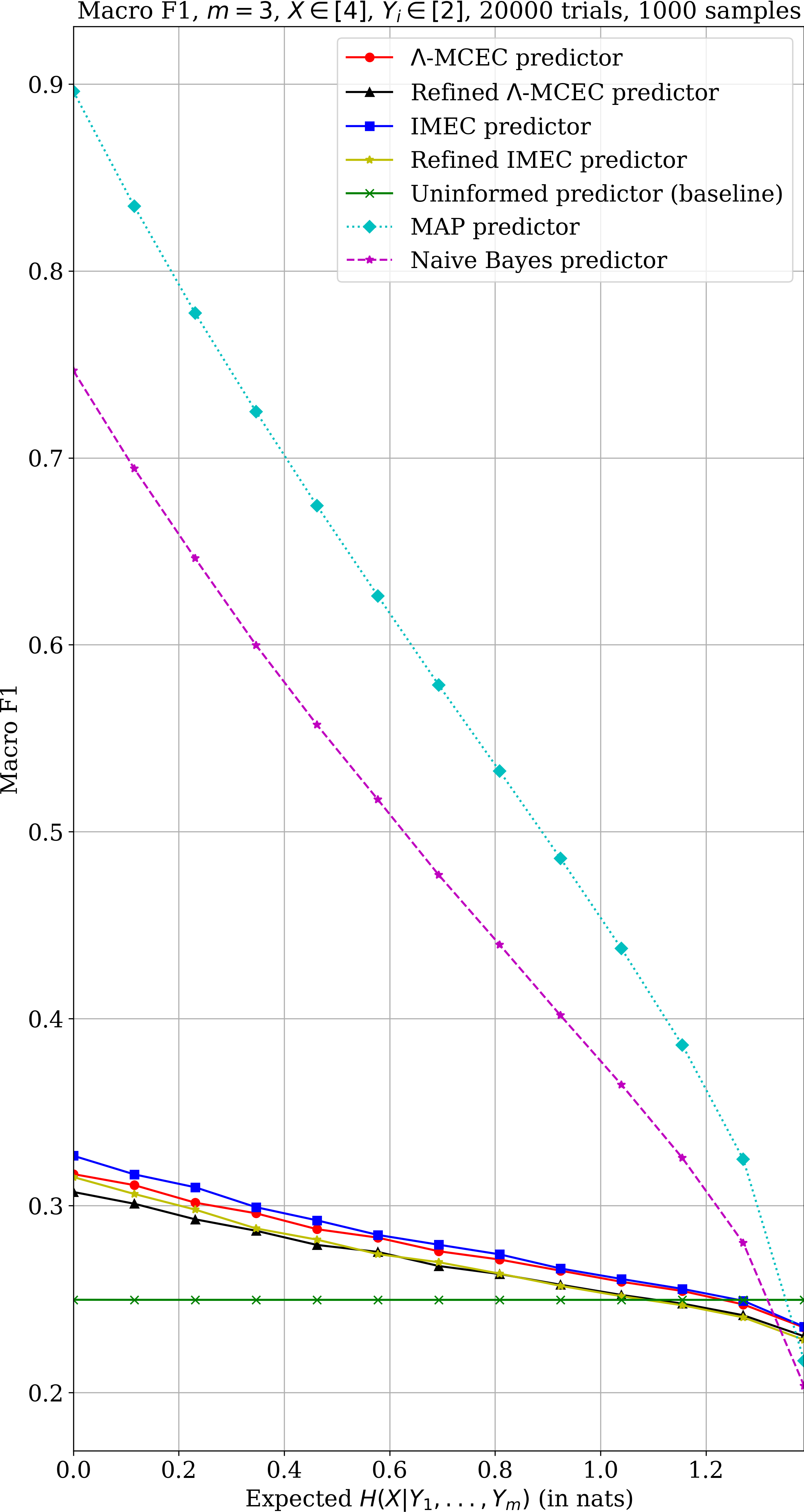}
\par\end{centering}
\caption{Comparison between various predictors, for $h\in[0,\log n_{X}]$ (expected value of $H(X|\mathbf{Y})$), and various choices of $m,n_{X},n_{Y}$ indicated in the titles of the subplots.}\label{fig:compare_more}
\end{figure*}

Figure \ref{fig:compare_more} shows that $\varLambda$-MCEC and IMEC is in the middle between MAP/naive Bayes and the uninformed predictor for $m=2,X\in [3], Y_i \in [2]$. Nevertheless, the gap between $\varLambda$-MCEC/IMEC and MAP/naive Bayes widens as $m$ becomes larger. For larger $m$, the problem of estimating the coupling given the marginals becomes more underdetermined, so it is expected that not having access to the coupling ($\varLambda$-MCEC, IMEC) will have a larger penalty compared to having access to the coupling (MAP/naive Bayes).

We remark that MAP and naive Bayes give smaller macro F1 than $\varLambda$-MCEC, IMEC and the uninformed predictor when the expected $H(X|\mathbf{Y})$ is maximized, where $X$ is independent of $\mathbf{Y}$. This is due to the imbalance of the portions of predicted classes given by MAP and naive Bayes. Although some randomization of the prediction might help improve the macro F1 of MAP and naive Bayes in this case, this is not usually performed in practice, so we do not include this.

\section{Details of Experiments with Datasets\label{app:datasets}}

We provide details for the experiments in Section \ref{subsec:datasets}.
For the Bank Marketing dataset \cite{moro2014data}, we only use categorical or binary feature variables, and ignore other variables (e.g., integers and dates). For the Soybean (large) dataset \cite{michalski1988soybean}, we discard the last $4$ classes due to them having insufficient samples.

Each data point in Figure~\ref{fig:compare_datasets} is computed via $10$-fold cross validation. The datasets may contain missing values in some features. During training, missing values are ignored when computing marginals $Q_{1},\ldots,Q_{m}$. During prediction, to compute $\hat{P}(x|y_{1},\ldots,y_{m})$ when $y_{i}$ are missing for $i\in E$ (where $E\subseteq[m]$), we compute the prediction as $\hat{P}(x|y_{1},\ldots,y_{m})=\mathbb{E}[\hat{P}(x|\hat{Y}_{1},\ldots,\hat{Y}_{m})]$, where $\hat{Y}_{i}=y_{i}$ if $i\notin E$, and $\hat{Y}_{i}\sim Q_{i}$ are independent for $i\in E$.

\section{Evaluation of MCEC Lower Bounds} \label{app:MCEC_bounds}

\begin{figure*}
\begin{centering}
\includegraphics[scale=0.29]{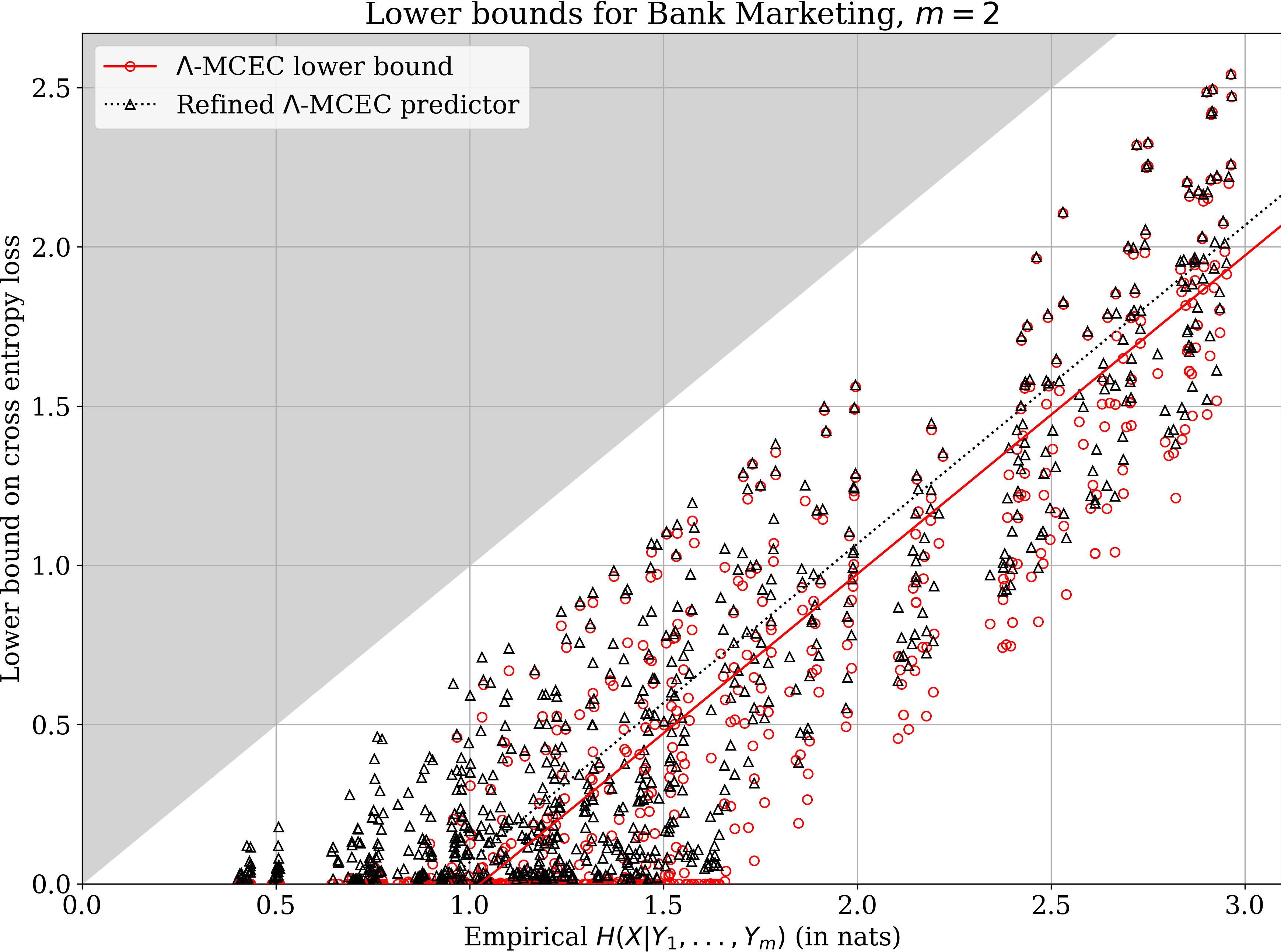}\includegraphics[scale=0.29]{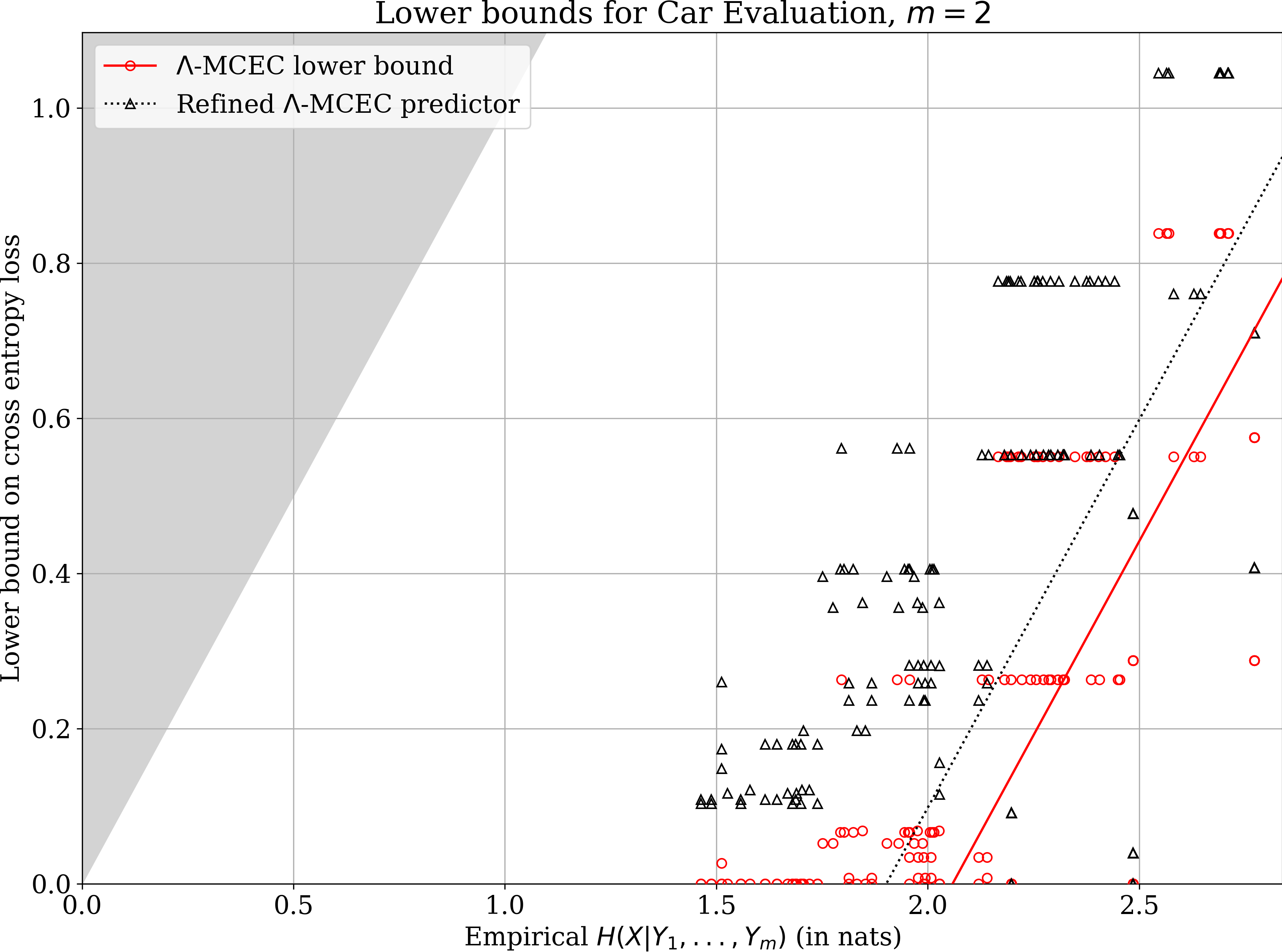}
\par\end{centering}
\medskip{}

\begin{centering}
\includegraphics[scale=0.29]{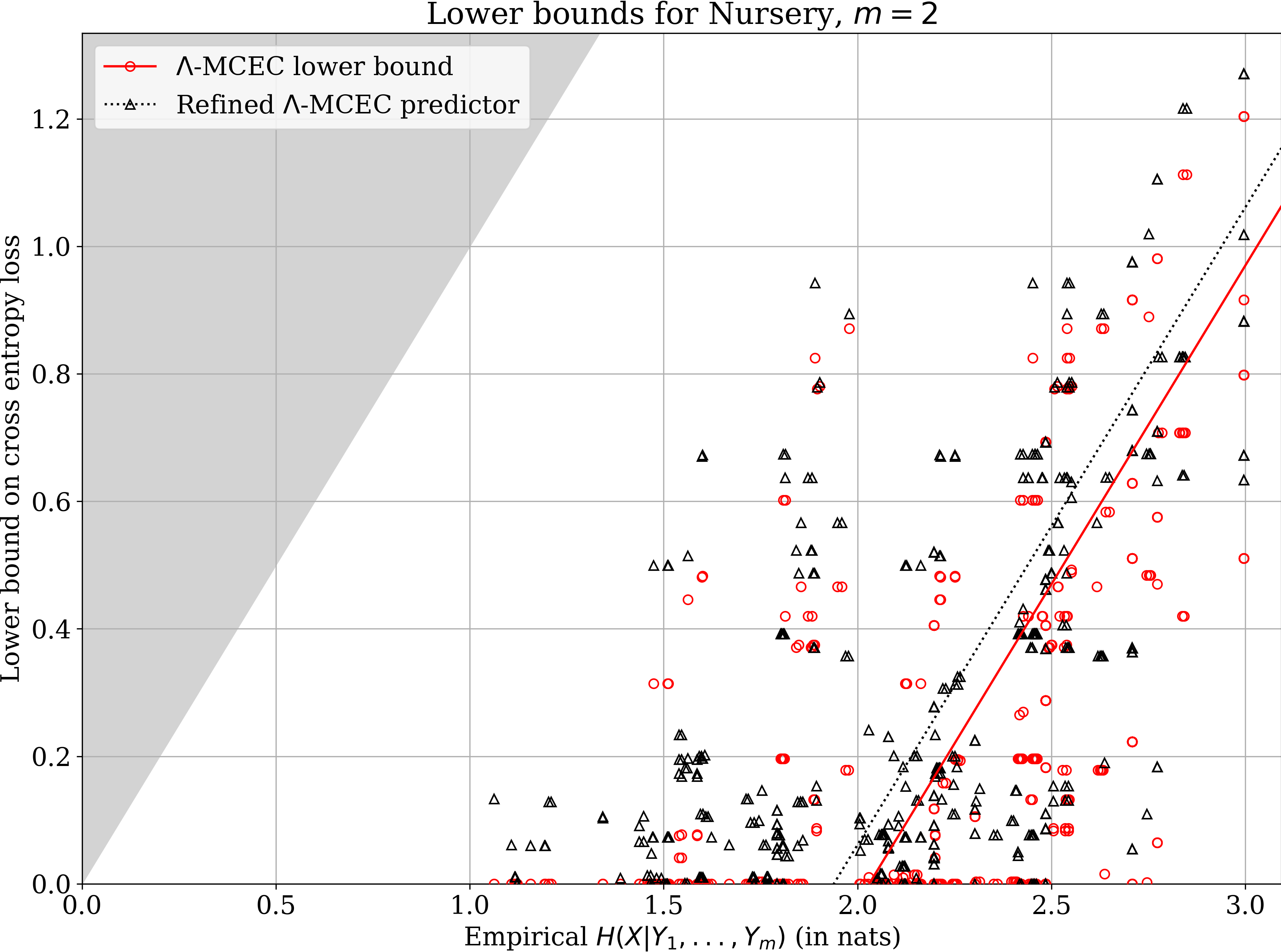}
\par\end{centering}
\caption{Comparison between $\varLambda$-MCEC lower bound and refined $\varLambda$-MCEC predictor for various datasets.}\label{fig:compare_datasets_bounds}
\end{figure*}

We evaluate the $\varLambda$-MCEC lower bound \eqref{eq:Lambda_MCEC_lb} on the datasets Bank Marketing \cite{moro2014data}, Car Evaluation \cite{bohanec1988knowledge} and Nursery \cite{olave1989application} (see Section \ref{subsec:datasets}; we only use these datasets since their number of features is relatively small). The $\varLambda$-MCEC lower bound is not effective when $H(P) \ll \sum_i H(Q_i)$ since the lower bound is often $0$ (i.e., perfect prediction is possible; see Section \ref{sec:MCEC_prop}). To study the scenario where the lower bound is applicable, we consider the following experiment where the class variable is a group of variables.

Let $m'$ be the number of (feature or class) variables in the dataset. We first pick $k=2$ variables to group into a single class variable (which takes values over tuples of $k$ values), and then pick $m=2$ other variables as the features $Y_1,\ldots,Y_m$. For each of the $\binom{m'}{k,m,m'-k-m}$ combinations, we evaluate the following quantities:
\begin{itemize}
\item The conditional entropy $H(X|Y_1,\ldots,Y_m)$ over the empirical distribution, which is the best possible empirical risk (average loss) for predicting $X$ using $Y_1,\ldots,Y_m$. It is computed over the whole empirical distribution of the dataset.
\item The $\varLambda$-MCEC lower bound, which is a lower bound on $H(X|Y_1,\ldots,Y_m)$ and the empirical risk computed using only the marginal empirical distributions of $X,Y_1,\ldots,Y_m$.
\item The average cross entropy loss of the refined $\varLambda$-MCEC predictor (evaluated over the coupling found by the refined $\varLambda$-MCEC algorithm), which is a lower bound of the best possible average cross entropy loss if the $\varLambda$-MCEC predictor is indeed the optimal coupling in \eqref{eq:mcec}. Otherwise, it is only an approximate lower bound according to Corollary \ref{cor:MCEC_almost_optimal}.
\end{itemize}
Each combination gives two points in the plots in Figure \ref{fig:compare_datasets_bounds}, where the $x$-value is the empirical conditional entropy, and the $y$-values are the $\varLambda$-MCEC lower bound and the empirical risk of the refined $\varLambda$-MCEC predictor.\footnote{Note that cross validation is unnecessary here since we are comparing the best possible empirical risk and the lower bounds over the whole dataset. There is no need for a separate testing data.}
The gray area in the plot is the region $y>x$, where the lower bound cannot lie in since it lower-bounds $H(X|Y_1,\ldots,Y_m)$.
We also plot two diagonal lines $x-y = \mathbb{E}[H(X|Y_1,\ldots,Y_m) - \mathrm{LB}]$, where the expectation is over a random combination among the $\binom{m'}{k,m,m'-k-m}$ combinations, and $\mathrm{LB}$ is the $\varLambda$-MCEC lower bound or the empirical risk of the refined $\varLambda$-MCEC predictor. Note that $\mathbb{E}[H(X|Y_1,\ldots,Y_m) - \mathrm{LB}]$ is the average gap between the empirical conditional distribution and the lower bound. The closeness of these lines to the gray area indicates the tightness of the bounds.

From Figure \ref{fig:compare_datasets_bounds}, we can see that the refined $\varLambda$-MCEC predictor generally gives a tighter bound, though it may not be a true lower bound if the $\varLambda$-MCEC predictor does not give the optimal coupling. Both bounds can be far from the empirical conditional entropy, which is expected since these bounds require much less information 
(only the marginal empirical distributions) to compute.

\end{document}